\documentclass[11pt,onecolumn]{IEEEtran}
\usepackage{cite}
\usepackage{graphicx}
\usepackage{amssymb}
\usepackage{amsfonts}
\usepackage[cmex10]{amsmath}
\usepackage{array}
\usepackage{mdwmath}

\newtheorem{theorem}{\bf Theorem}

\newtheorem{corollary}{\bf Corollary}

\newcommand{\xh}{\mathbf{h}}
\newcommand{\xH}{\mathbf{H}}
\newcommand{\xX}{\mathbf{X}}

\newcommand{\xY}{\mathbf{Y}}
\newcommand{\xV}{\mathbf{V}}
\newcommand{\xZ}{\mathbf{Z}}
\newcommand{\xv}{\mathbf{v}}
\newcommand{\barxV}{\mathbf{\bar{V}}}
\newcommand{\barxH}{\mathbf{\bar{H}}}
\newcommand{\barxX}{\mathbf{\bar{X}}}
\newcommand{\barxY}{\mathbf{\bar{Y}}}
\newcommand{\barxZ}{\mathbf{\bar{Z}}}
\newcommand{\barxv}{\mathbf{\bar{v}}}

\begin{document}

\title{Degrees of Freedom of the $K$ User $M \times N$ MIMO Interference Channel}

\author{\authorblockN{Tiangao Gou, Syed A. Jafar}\\
\authorblockA{Electrical Engineering and Computer Science\\
University of California Irvine, Irvine, California, 92697, USA\\
Email: \{tgou,syed\}@uci.edu}}\maketitle \IEEEpeerreviewmaketitle

\begin{abstract}
We provide innerbound and outerbound for the total number of degrees
of freedom of the $K$ user multiple input multiple output (MIMO)
Gaussian interference channel with $M$ antennas at each transmitter
and $N$ antennas at each receiver if the channel coefficients are
time-varying and drawn from a continuous distribution. The bounds
are tight when the ratio $\frac{\max(M,N)}{\min(M,N)}=R$ is equal to
an integer. For this case, we show that the total number of degrees
of freedom is equal to $\min(M,N)K$ if $K \leq R$ and
$\min(M,N)\frac{R}{R+1}K$ if $K
> R$. Achievability is based on interference alignment. We
also provide examples where using interference alignment combined
with zero forcing can achieve more degrees of freedom than merely
zero forcing for some MIMO interference channels with constant
channel coefficients.
\end{abstract}
\newpage
\section{introduction}
Interference management is an important problem in wireless system
design. Researchers have been exploring the capacity
characterization of the Gaussian interference channel from a
information theoretic perspective for more than thirty years.
Several innerbounds and outerbounds of the capacity region for the
two user Gaussian interference channel with single antenna nodes are
determined
\cite{Carleial:75IT,Sato:81IT,Han&Kobayashi:81IT,Costa:85IT,Sato:77IT,Carleial:83IT,Kramer:04IT,Etkin-etal:07IT_submission,Telatar&Tse:07ISIT,Shang-etal:06IT_submission}.
However, the capacity region of the Gaussian interference channel
remains an open problem in general. Interference channels with
multiple-antenna nodes are studied in
\cite{Vishwanath-Jafar:ITW,Shang-etal:MIMO,ParetoMISO}.

\subsection{Motivating Example}
In \cite{ParetoMISO}, the authors study the achievable rate region
of the multiple input single output (MISO) interference channel
obtained by treating interference as noise. They parameterize the
Pareto boundary of the MISO Gaussian interference channel for
arbitrary number of users and antennas at the transmitter as long as
the number of antennas is larger than the number of users. For 2
user case, they show that the optimal beamforming directions are a
linear combination of maximum ratio transmission vectors and the
zero forcing vectors. However, for the case when the number of
antennas is less than that of users, the optimal beamforming
direction is not known. Intuitively, this is because when the number
of antennas is less than that of users, it is not possible for each
user to choose beamforming vectors to ensure no interference is
created at all other users. The same problem is evident when we
study this channel from a degrees of freedom \footnote{If the sum
capacity can be expressed as $C_{\Sigma}(SNR)=\eta
\log(SNR)+o(\log(SNR))$ then we say that the channel has $\eta$
degrees of freedom.} perspective. For the 2 user MISO interference
channel with 2 transmit antennas and a single receive antenna, it is
easy to see 2 degrees of freedom can be achieved if each user
chooses zero forcing beamforming vector so that no interference is
created at the other user. This is also the maximum number of
degrees of freedom of this channel. However, for 3 user MISO
interference channel with two antennas at each transmitter, it is
not possible for each user to choose beamforming vectors so that no
interference is created at all other users. As a result, only 2
degrees of freedom can be achieved by zero forcing. Can we do better
than merely zero forcing? What is the total number of degrees of
freedom of the 3 user MISO interference channel with 2 antennas at
each transmitter? In general, what is the total number of degrees of
freedom of the $K$ user $M \times N$ MIMO interference channel?
These are the questions that we explore in this paper.

Before we answer the above questions, let us first review the
results on the degrees of freedom for the $K$ user single input
single output (SISO) Gaussian interference channel. If $K=1$, it is
well known the degrees of freedom for this point to point channel is
1. If $K=2$, it is shown that this channel has only 1 degrees of
freedom \cite{Nosratinia-Madsen}. In other words, each user can
achieve $\frac{1}{2}$ degrees of freedom simultaneously. For $K>2$,
it is surprising that every user is still able to achieve
$\frac{1}{2}$ degrees of freedom no matter how large $K$ is, if the
channel coefficients are time-varying or frequency selective and
drawn from a continuous distribution \cite{Cadambe_Jafar_int}. The
achievable scheme is based on interference alignment combined with
zero forcing.

For the MISO interference channel we find a similar characterization
of the degrees of freedom. For example, the degrees of freedom for
the 3 user MISO interference channel with 2 antennas at each
transmitter is only 2 which is the same as that for the 2 user case.
In other words, every user can achieve $\frac{2}{3}$ degrees of
freedom simultaneously. For $K>3$, every user is still able to
achieve $\frac{2}{3}$ degrees of freedom regardless of $K$ if the
channel coefficients are time-varying or frequency selective and
drawn from a continuous distribution. The achievable scheme is based
on interference alignment on the single input multiple output (SIMO)
interference channel for simplicity. If interference alignment is
achieved on the SIMO channel it can also be achieved on the MISO
channel, due to a reciprocity of alignment
\cite{Gomadam_Cadambe_Jafar_dist}. Interestingly, the interference
alignment scheme is different from all prior schemes. All prior
interference alignment schemes \cite{Cadambe_Jafar_int} (including
the ones for the $X$ channel \cite{Jafar_Shamai,Cadambe_Jafar_X})
explicitly achieve one-to-one alignment of signal vectors, i.e., to
minimize the dimension of the space spanned by interference signal
vectors, one signal vector from an interferer and one signal vector
from another interferer are aligned along the same dimension at the
desired receivers. For example, consider 3 user SISO interference
channel with 2 symbol extension or 3 user MIMO interference channel
where each node has 2 antennas. We need to choose beamforming
vectors $\mathbf{v}^{[2]}$ and $\mathbf{v}^{[3]}$ at Transmitter 2
and 3, respectively so that they cast overlapping shadow at Receiver
1, i.e.,
\begin{displaymath}
\mathbf{H}^{[12]}\mathbf{v}^{[2]}=\mathbf{H}^{[13]}\mathbf{v}^{[3]}
\end{displaymath}
where $\mathbf{H}^{[12]}$ and $\mathbf{H}^{[13]}$ are $2 \times 2$
channel matrices from Transmitter 2 and 3 to Receiver 1,
respectively. However, such an alignment is not feasible on the SIMO
channel. Notice that the solution to the condition mentioned above
exists only when the range of the two channel matrices has
intersection. The channel matrix for 2 symbol extension SIMO channel
with 2 antennas at each receiver is $4 \times 2$. The range of two
such channel matrices has null intersection with probability one if
the channel coefficients are drawn from a continuous distribution.
Thus, one-to-one interference alignment does not directly work for
SIMO channel. Instead, interference from one interferer can only be
aligned within the union of the spaces spanned by the interference
vectors from $R$ other interferers where $R$ is the number of
antennas at each receiver.

\subsection{Overview of Results}
In this paper we study the degrees of freedom of the $K$ user MIMO
Gaussian interference channel with $M$ antennas at each transmitter
and $N$ antennas at each receiver. We provide both the innerbound
(achievability) and outerbound (converse) of the total number of
degrees of freedom for this channel. We show that $\min(M,N)K$
degrees of freedom can be achieved if $K \leq R$ and
$\frac{R}{R+1}\min(M,N)K$ degrees of freedom can be achieved if
$K>R$ where $R=\lfloor\frac{\max(M,N)}{\min{(M,N)}}\rfloor$. The
total number of degrees of freedom is bounded above by $\min(M,N)K$
if $K \leq R$ and $\frac{\max(M,N)}{R+1}K$ if $K>R$. The bounds are
tight when the ratio $\frac{\max(M,N)}{\min(M,N)}=R$ is equal to an
integer which includes MISO and SIMO interference channel as special
cases. The result indicates when $K \leq R$ every user can achieve
$\min(M,N)$ degrees of freedom which is the same as what one can
achieve without interference. When $K>R$ every user can achieve a
fraction  $\frac{R}{R+1}$ of the degrees of freedom that one can
achieve in the absence of all interference. In other words, if $K
\leq R$, then there is no loss of degrees of freedom for each user
with interference. If $K > R$, every user only loses a fraction
$\frac{1}{R+1}$ of the degrees of freedom that can be achieved
without interference. In the second part of this paper we study the
achievable degrees of freedom based on interference alignment scheme
for the $R+2$ user MIMO interference channel with $M$ antennas at
each transmitter and $RM$, $R=2,3,\ldots$ antennas at each receiver
and constant channel coefficients, i.e. in the absence of time
variation. We show that for this channel
$RM+\lfloor\frac{RM}{R^2+2R-1}\rfloor$ degrees of freedom can be
achieved without symbol extension. When
$\lfloor\frac{RM}{R^2+2R-1}\rfloor<0$ and hence $M<R+2$,
$RM+\frac{1}{\lceil\frac{R+2}{M}\rceil}$ degrees of freedom per
orthogonal dimension can be achieved with finite symbol extension.
Since only $RM$ degrees of freedom can be achieved using zero
forcing, these results provide interesting examples where using
interference alignment scheme can achieve more degrees of freedom
than merely zero forcing.

\section{system model}
The $K$ user MIMO interference channel is comprised of $K$
transmitters and $K$ receivers. Each transmitter has $M$ antennas
and each receiver has $N$ antennas. The channel output at the
$k^{th}$ receiver over the $t^{th}$ time slot is characterized by
the following input-output relationship:
\begin{equation*}
\mathbf{Y}^{[k]}(t)=\xH^{[k1]}(t)\xX^{[1]}(t)+\xH^{[k2]}(t)\xX^{[2]}(t)+\cdots+\xH^{[kK]}(t)\xX^{[K]}(t)+\xZ^{[k]}(t)
\end{equation*}
where, $k\in \{1,2,\cdots,K\}$ is the user index, $t \in \mathbb{N}$
is the time slot index, $\mathbf{Y}^{[k]}(t)$ is the $N \times 1$
output signal vector of the $k^{th}$ receiver, $\xX^{[j]}(t)$ is the
$M \times 1$ input signal vector of the $j^{th}$ transmitter,
$\xH^{[kj]}(t)$ is the $N \times M$ channel matrix from transmitter
$j$ to receiver $k$ over the $t^{th}$ time slot and
$\mathbf{Z}^{[k]}(t)$ is $N\times 1$ additive white Gaussian noise
(AWGN) vector at the $k^{th}$ receiver. We assume all noise terms
are i.i.d zero mean complex Gaussian with unit variance. We assume
that all channel coefficient values are drawn i.i.d. from a
continuous distribution and the absolute value of all the channel
coefficients is bounded between a non-zero minimum value and a
finite maximum value. The channel coefficient values vary at every
channel use. Perfect knowledge of all channel coefficients is
available to all transmitters and receivers.

Transmitters $1, 2, \cdots, K$ have independent messages $W_1, W_2,
\cdots, W_K$ intended for receivers $1, 2, \cdots, K$, respectively.
The total power across all transmitters is assumed to be equal to
$\rho$. We indicate the size of the message set by $|W_i(\rho)|$.
For codewords spanning $t_0$ channel uses, the rates
$R_i(\rho)=\frac{\log|W_i(\rho)|}{t_0}$ are achievable if the
probability of error for all messages can be simultaneously made
arbitrarily small by choosing an appropriately large $t_0$. The
capacity region $\mathcal{C}(\rho)$ of the $K$ user MIMO
interference channel is the set of all achievable rate tuples
${\mathbf R}(\rho)=(R_1(\rho), R_2(\rho), \cdots, R_K(\rho))$.

We define the spatial degrees of freedom as:
\begin{equation}
\eta \triangleq \lim_{\rho \rightarrow \infty}
\frac{C_\Sigma(\rho)}{\log(\rho)}
\end{equation}
where $C_\Sigma(\rho)$ is the sum capacity at SNR $\rho$.

\section{Outerbound on the degrees of freedom for the $K$ user MIMO interference channel}
We provide an outerbound on the degrees of freedom for the $K$ user
 MIMO Gaussian interference channel in this section. Note that the converse holds for both time-varying and
 constant (non-zero)
channel coefficients, i.e., time variations are not required. We
present the result in the following theorem:
\begin{theorem}\label{thm:outerbound}
For the $K$ user MIMO Gaussian interference channel with $M$
antennas at each transmitter and $N$ antennas at each receiver, the
total number of degrees of freedom is bounded above by $K\min(M,N)$
if $K \leq R$ and $\frac{\max(M,N)}{R+1}K$ if $K>R$ where
$R=\lfloor\frac{\max(M,N)}{\min{(M,N)}}\rfloor$, i.e.
\begin{equation*}
\eta=d_1+\cdots+d_K \leq \min{(M,N)}K~1(K \leq
R)+\frac{\max(M,N)}{R+1}K~1(K>R)
\end{equation*}
where 1(.) is the indicator function and $d_i$ represents the
individual degrees of freedom achieved by user $i$.
\end{theorem}
\begin{proof}\\
1) $K \leq R$: It is well known that the degrees of freedom of a
single user MIMO Gaussian channel with $M$ transmit antennas and $N$
receive anteanns is equal to $\min(M,N)$. Thus, for the $K$ user
MIMO Gaussian interference channel with the same antenna deployment,
the degrees of freedom cannot be more than $K\min(M,N)$, i.e $\eta
\leq
K\min(M,N)$.\\
2) $K>R$: Consider the $R+1$ user MIMO interference channel with
$M,N$ antennas at the transmitter and receiver respectively. If we
allow full cooperation among $R$ transmitters and full cooperation
among their corresponding receivers, then it is equivalent to the
two user MIMO interference channel with $RM$, $M$ (respectively)
antennas at transmitters and $RN$, $N$ antennas at their
corresponding receivers. In \cite{Jafar_dof_int}, it is shown that
the degrees of freedom for a two user MIMO Gaussian interference
channel with $M_1$, $M_2$ antennas at transmitter $1$, $2$ and
$N_1$, $N_2$ antennas at their corresponding receivers is
min\{$M_1+M_2$, $N_1+N_2$, max($M_1$,$N_2$), max($M_2$,$N_1$)\}.
From this result, the degrees of freedom for the two user MIMO
interference channel with $RM$, $M$ antennas at the transmitters and
$RN$, $N$ at their corresponding receivers is $\max(M,N)$. Since
allowing transmitters and receivers to cooperate does not hurt the
capacity, the degrees of freedom of the original $R+1$ user
interference channel is no more than $\max(M,N)$. For $K>R+1$ user
case, picking any $R+1$ users among $K$ users gives an outerbound:
\begin{equation}
d_{i_1}+d_{i_2}+\cdots+d_{i_{R+1}} \leq \max(M,N) \quad \forall
i_1,\cdots,i_{R+1} \in \{1,2,\cdots,K\}, \quad i_1 \ne i_2 \ne
\cdots \ne i_{R+1}
\end{equation}
Adding up all such inequalities, we get the outerbound of the $K$
user MIMO interference channel:
\begin{equation}
d_1+d_2+\cdots+d_K \leq \frac{\max(M,N)}{R+1}K
\end{equation}
\end{proof}

\section{Innerbound on the degrees of freedom for the $K$ user MIMO interference channel}
To derive the innerbound on the degrees of freedom for the $K$ user
MIMO Gaussian interference channel, we first obtain the achievable
degrees of freedom for the $K$ user SIMO interference channel with
$R$ antennas at each receiver. The innerbound on the degrees of
freedom of the $K$ user MIMO interference channel follows directly
from the results of the SIMO interference channel. The corresponding
input-output relationship of the $K$ user SIMO interference channel
is:
\begin{equation*}
\mathbf{Y}^{[k]}(t)=\xh^{[k1]}(t)x^{[1]}(t)+\xh^{[k2]}(t)x^{[2]}(t)+\cdots+\xh^{[kK]}(t)x^{[K]}(t)+\xZ^{[k]}(t)
\end{equation*}
where $\mathbf{Y}^{[k]}(t)$,  $x^{[j]}(t)$, $\xh^{[kj]}(t)$,
$\xZ^{[k]}(t)$ represent the channel output at receiver $k$, the
channel input from transmitter $j$, the channel vector from
transmitter $j$ to receiver $k$ and the AWGN vector at receiver $k$
over the $t^{th}$ time slot respectively.

We start with the problem mentioned in the introduction. For the 3
user SIMO Gaussian interference channel with 2 receive antennas, 2
degrees of freedom can be achieved using zero forcing. From the
converse result in the last section, we cannot achieve more than 2
degrees of freedom on this channel. Therefore, the maximum number of
degrees of freedom for this channel is 2. For the 4 user case, the
converse result indicates that this channel cannot achieve more than
$\frac{8}{3}$ degrees of freedom. Can we achieve this outerbound?
Interestingly, using interference alignment scheme based on
beamforming over multiple symbol extensions of the original channel,
we are able to approach arbitrarily close to the outerbound.
Consider the  $\mu_n=3(n+1)^8$ symbol extension of the channel for
any arbitrary $n \in \mathbb{N}$. Then, we effectively have a
$2\mu_n \times \mu_n$ channel with a block diagonal structure. In
order for each user to get exactly $\frac{2}{3}$ degrees of freedom
per channel use and hence $\frac{2}{3}\mu_n=2(n+1)^8$ degrees of
freedom on the $\mu_n$ symbol extension channel, each receiver with
a total of $2\mu_n$ dimensional signal space should partition its
signal space into two disjoint subspaces, one of which has
$\frac{2}{3}\mu_n$ dimension for the desired signals and the other
has $\frac{4}{3}\mu_n$ dimension for the interference signals. While
such an alignment would exactly achieve the outerbound, it appears
to be infeasible in general. But if we allow user 4 to achieve only
$(\frac{2}{3}-\epsilon_n)\mu_n=2n^8$ degrees of freedom over the
$\mu_n$ extension channel where
$\epsilon_n=\frac{2(n+1)^8-2n^8}{3(n+1)^8}=\frac{2}{3}[1-\frac{1}{(1+\frac{1}{n})^8)}]$,
then it is possible for user 1, 2, 3 to achieve exactly
$\frac{2}{3}\mu_n$ degrees of freedom simultaneously for a total of
$(\frac{8}{3}-\epsilon_n)\mu_n$ degrees of freedom over the $\mu_n$
symbol extension channel. Hence,
$\frac{8}{3}-\frac{2}{3}[1-\frac{1}{(1+\frac{1}{n})^8)}]$ degrees of
freedom per channel use can be achieved. As $n \to \infty$,
$\frac{2}{3}[1-\frac{1}{(1+\frac{1}{n})^8)}] \to 0$. Therefore, we
can achieve arbitrarily close to the outerbound $\frac{8}{3}$. Next
we present a detailed description of the interference-alignment
scheme for the 4 user SIMO channel with 2 antennas at each receiver.

In the extended channel, Transmitter $j, \forall j=1,2,3$ sends
message $W_j$ to Receiver $j$ in the form of $\frac{2}{3}\mu_n$
independently encoded steams $x^{[j]}_m(t),
m=1,2,\ldots,\frac{2}{3}\mu_n$ along the same set of beamforming
vectors
$\barxv^{[1]}_1(t),\ldots,\barxv^{[1]}_{\frac{2}{3}\mu_n}(t)$, each
of dimension $\mu_n \times 1$, so that we have
\begin{figure}[t]
\centering
\includegraphics[width=6.4in]{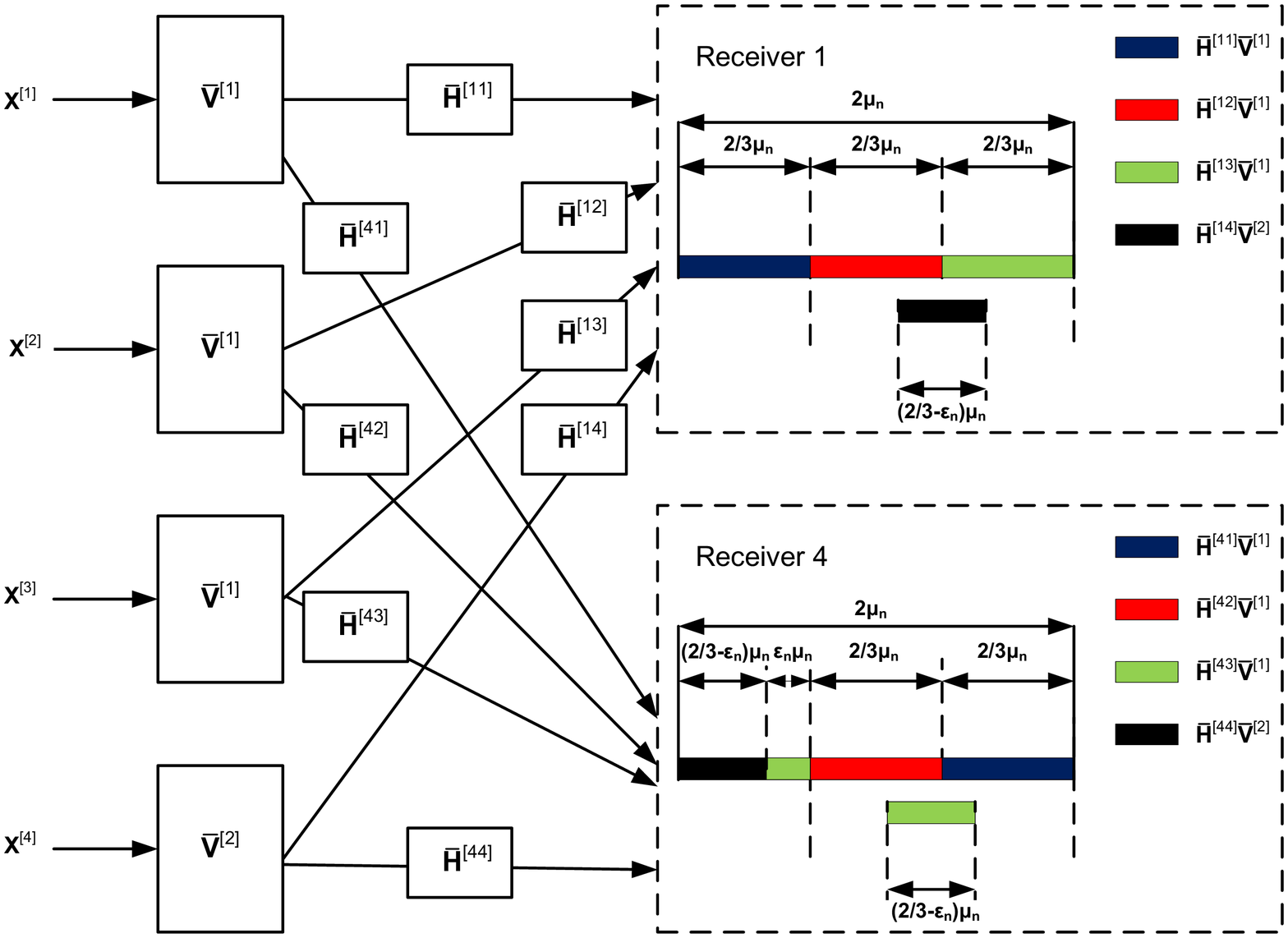}
\caption{Interference alignment on the 4 user interference channel}
\label{fig1}
\end{figure}
\begin{equation*}
\barxX^{[j]}(t) = \displaystyle\sum_{m=1}^{\frac{2}{3}\mu_n}
x^{[j]}_m(t) \barxv_m^{[1]}(t) = \barxV^{[1]}(t) \xX^{[j]}(t), ~~~ j
= 1,2,3
\end{equation*}
where
$\barxV^{[1]}(t)=[\barxv^{[1]}_1(t),\cdots,\barxv^{[1]}_{\frac{2}{3}\mu_n}(t)]$
is a $\mu_n \times \frac{2}{3}\mu_n$ matrix and $\xX^{[j]}(t)$ is a
$\frac{2}{3}\mu_n \times 1$ column vector. Transmitter 4 sends
message $W_4$ to Receiver 4 in the form of
$(\frac{2}{3}-\epsilon_n)\mu_n$ independently encoded streams
$x^{[4]}_m(t), m=1,2,\ldots,(\frac{2}{3}-\epsilon_n)\mu_n$ along the
beamforming vectors
$\barxv^{[2]}_1(t),\ldots,\barxv^{[2]}_{(\frac{2}{3}-\epsilon_n)\mu_n}(t)$
so that
\begin{equation*}
\barxX^{[4]}(t) =
\displaystyle\sum_{m=1}^{(\frac{2}{3}-\epsilon_n)\mu_n} x^{[4]}_m(t)
\barxv_m^{[2]}(t) = \barxV^{[2]}(t) \xX^{[4]}(t)
\end{equation*}
where
$\barxV^{[2]}(t)=[\barxv^{[2]}_1(t),\cdots,\barxv^{[2]}_{(\frac{2}{3}-\epsilon_n)\mu_n}(t)]$
is a $\mu_n \times (\frac{2}{3}-\epsilon_n)\mu_n$ matrix and
$\xX^{[4]}(t)$ is a $(\frac{2}{3}-\epsilon_n)\mu_n \times 1$ column
vector. Therefore, the received signal at Receiver $k$ is
\begin{equation*}
\barxY^{[k]}(t) = \displaystyle\sum_{j=1}^{3}\barxH^{[kj]}(t)
\barxV^{[1]}(t)\xX^{[j]}(t) +
\barxH^{[k4]}(t)\barxV^{[2]}(t)\xX^{[4]}(t) +
\bar{\mathbf{Z}}^{[k]}(t)
\end{equation*}
where $\barxH^{[kj]}(t)$ is the $2\mu_n \times \mu_n$ matrix
representing the $\mu_n$ extension of the original channel matrix,
i.e.
\begin{eqnarray*}
\bar{\mathbf{H}}^{[kj]}(t)=  \left[ \begin{array}{cccc}  \mathbf{h}^{[kj]}(\mu_n(t-1)+1) & \mathbf{0} & \ldots & \mathbf{0}\\
     \mathbf{0} & \mathbf{h}^{[kj]}(\mu_n(t-1)+2) & \ldots & \mathbf{0}\\
    \vdots & \cdots & \ddots & \vdots\\
     \mathbf{0} & \mathbf{0}& \cdots  & \mathbf{h}^{[kj]}(\mu_nt) \\
    \end{array}\right]
\end{eqnarray*}
where $\mathbf{0}$ is a $2 \times 1$ vector with zero entries.
Similarly, $\barxY$ and $\bar{\mathbf{Z}}$ represent the $\mu_n$
symbol extension of the $\xY$ and $\xZ$ respectively. The
interference alignment scheme is shown in Fig. \ref{fig1}. At
Receiver 1, the interference from Transmitter 2 and Transmitter 3
cannot be aligned with each other because the subspaces spanned by
the columns of $\barxH^{[12]}$ and $\barxH^{[13]}$ have null
intersection with probability one. Thus, the interference vectors
from Transmitter 2, i.e. columns of $\barxH^{[12]}\barxV^{[1]}$ and
interference vectors from Transmitter 3, i.e. columns of
$\barxH^{[13]}\barxV^{[1]}$ together span a $\frac{4}{3}\mu_n$
dimensional subspace in the $2\mu_n$ dimensional signal space at
Receiver 1. In order for Receiver 1 to get a $\frac{2}{3}\mu_n$
dimensional interference-free signal space, we need to align the
space spanned by the interference vectors from Transmitter 4, i.e.
the range of $\barxH^{[14]}\barxV^{[2]}$ within the space spanned by
the interference vectors from Transmitter 2 and 3. Note that we
cannot align the interference from Transmitter 4 within the space
spanned by the interference vectors from Transmitter 2 only or
Transmitter 3 only. Because the subspaces spanned by the columns of
$\barxH^{[14]}$ and $\barxH^{[12]}$ or the subspaces spanned by the
columns of $\barxH^{[14]}$ and $\barxH^{[13]}$ have null
intersection with probability one. Mathematically, we have
\begin{equation}\label{reqrx1}
\text{span}(\barxH^{[14]}\barxV^{[2]}) \subset \text{span}(\left[
\barxH^{[12]}\barxV^{[1]}~ \barxH^{[13]}\barxV^{[1]} \right])
\end{equation}
where $\text{span}(\mathbf{A})$ means the space spanned by the
columns of matrix $\mathbf{A}$. This condition can be expressed
equivalently as
\begin{equation*}
\text{span}(\barxH^{[14]}\barxV^{[2]}) \subset \text{span} (\left[
\barxH^{[12]}~ \barxH^{[13]}\right]
\left[ \begin{array}{cc} \barxV^{[1]}& \mathbf{0}\\
\mathbf{0}&\barxV^{[1]} \end{array} \right])
\end{equation*}
where $\mathbf{0}$ denotes a $\mu_n \times \frac{2}{3}\mu_n$ matrix
with zero entries. Note that $[\barxH^{[12]}~ \barxH^{[13]}]$ is a
$2\mu_n \times 2\mu_n$ matrix with full rank almost surely.
Therefore, the last equation is equivalent to
\begin{eqnarray}\label{reqrx1_2}
\text{span}(\underbrace{[\barxH^{[12]}~\barxH^{[13]}]^{-1}\barxH^{[14]}}_{\mathbf{T}^{[1]}}
\barxV^{[2]}) \subset \text{span}(\left[ \begin{array}{cc} \barxV^{[1]}& \mathbf{0}\\
\mathbf{0}&\barxV^{[1]} \end{array} \right])
\end{eqnarray}
where $\mathbf{T}^{[1]}$ is a $2\mu_n \times \mu_n$ matrix which can
be written in a block matrix form:
\begin{equation*}
\mathbf{T}^{[1]}=\left[\begin{array}{c}\mathbf{T}^{[1]}_{1}\\
\mathbf{T}^{[1]}_{2}\end{array}\right]
\end{equation*}
where $\mathbf{T}^{[1]}_{1}$ and $\mathbf{T}^{[1]}_{2}$ are $\mu_n
\times \mu_n$ matrices. Therefore, \eqref{reqrx1_2} can be expressed
alternatively as
\begin{eqnarray}
\text{span}( \left[ \begin{array}{c} \mathbf{T}^{[1]}_1 \barxV^{[2]}\\
\mathbf{T}^{[1]}_2 \barxV^{[2]}\\ \end{array} \right] )\subset
\text{span}(\left[ \begin{array}{cc} \barxV^{[1]}& \mathbf{0}\\
\mathbf{0}& \barxV^{[1]} \end{array} \right])
\end{eqnarray}
This condition can be satisfied if
\begin{eqnarray}\label{cond1}
\left\{\begin{array}{ccc}
\mathbf{T}^{[1]}_1 \barxV^{[2]} & \prec & \barxV^{[1]} \\
\mathbf{T}^{[1]}_2 \barxV^{[2]} & \prec & \barxV^{[1]}  \\
\end{array}\right .
\end{eqnarray}
where $\mathbf{P}  \prec  \mathbf{Q}$  means that the set of column
vectors of matrix $\mathbf{P}$ is a subset of the set of column
vectors of matrix $\mathbf{Q}$.

Similarly, at Receiver 2, the interference vectors from Transmitter
4 are aligned within the space spanned by the interference vectors
from Transmitter 1 and 3, i.e.,
\begin{equation}
\text{span}(\barxH^{[24]}\barxV^{[2]}) \subset \text{span}(\left[
\barxH^{[21]}\barxV^{[1]}~ \barxH^{[23]}\barxV^{[1]} \right])
\end{equation}
This condition can be satisfied if
\begin{eqnarray}\label{cond2}
\left\{\begin{array}{ccc}
\mathbf{T}^{[2]}_1 \barxV^{[2]} & \prec & \barxV^{[1]} \\
\mathbf{T}^{[2]}_2 \barxV^{[2]} & \prec & \barxV^{[1]}  \\
\end{array}\right .
\end{eqnarray}
where
\begin{equation*}
\mathbf{T}^{[2]}=\left[\begin{array}{c}\mathbf{T}^{[2]}_{1}\\
\mathbf{T}^{[2]}_{2}\end{array}\right]=[\barxH^{[21]}~\barxH^{[23]}]^{-1}\barxH^{[24]}
\end{equation*}
At Receiver 3, the interference vectors from Transmitter 4 are
aligned within the space spanned by the interference vectors from
Transmitter 1 and 2, i.e.
\begin{equation}
\text{span}(\barxH^{[34]}\barxV^{[2]}) \subset \text{span}(\left[
\barxH^{[31]}\barxV^{[1]}~ \barxH^{[32]}\barxV^{[1]} \right])
\end{equation}
This condition can be satisfied if
\begin{eqnarray}\label{cond3}
\left\{\begin{array}{ccc}
\mathbf{T}^{[3]}_1 \barxV^{[2]} & \prec & \barxV^{[1]} \\
\mathbf{T}^{[3]}_2 \barxV^{[2]} & \prec & \barxV^{[1]}  \\
\end{array}\right .
\end{eqnarray}
where
\begin{equation*}
\mathbf{T}^{[3]}=\left[\begin{array}{c}\mathbf{T}^{[3]}_{1}\\
\mathbf{T}^{[3]}_{2}\end{array}\right]=[\barxH^{[31]}~\barxH^{[32]}]^{-1}\barxH^{[34]}
\end{equation*}
Now, let us consider Receiver 4. As shown in Fig. \ref{fig1}, to get
a $(\frac{2}{3}-\epsilon_n)\mu_n$ interference free dimensional
signal space, the dimension of the space spanned by the interference
vectors has to be less than or equal to
$2\mu_n-(\frac{2}{3}-\epsilon_n)\mu_n$. To achieve this, we align
the space spanned by $(\frac{2}{3}-\epsilon_n)\mu_n$ vectors of the
interference vectors from Transmitter 3 within the space spanned by
the interference from Transmitter 1 and 2. Since $\barxV^{[1]}$ is a
$\mu_n \times \frac{2}{3}\mu_n$ matrix, we can write it as
$\barxV^{[1]}=[\barxV^{[1]}_u~ \barxV^{[1]}_{\epsilon_n}]$ where
$\barxV^{[1]}_u$ and $\barxV^{[1]}_{\epsilon_n}$ are $\mu_n \times
(\frac{2}{3}-\epsilon_n)\mu_n$ and $\mu_n \times \epsilon_n\mu_n$
matrices, respectively. We assume the space spanned by the columns
of $\barxH^{[43]}\barxV^{[1]}_u$ is aligned within the space spanned
by the interference from Transmitter 1 and 2, i.e.,
\begin{equation}\label{alignrx4}
\text{span}(\barxH^{[43]}\barxV^{[1]}_u) \subset \text{span} (\left[
\barxH^{[41]}\barxV^{[1]}~ \barxH^{[42]}\barxV^{[1]} \right])
\end{equation}
From equation \eqref{cond1}, we have
\begin{equation*}
\mathbf{T}^{[1]}_1 \barxV^{[2]} \prec \barxV^{[1]}
\end{equation*}
This implies that $(\frac{2}{3}-\epsilon_n)\mu_n$ columns of
$\barxV^{[1]}$ are equal to the columns of $\mathbf{T}^{[1]}_1
\barxV^{[2]}$. Without loss of generality, we assume that
$\barxV^{[1]}_u=\mathbf{T}^{[1]}_1 \barxV^{[2]}$. Thus,
\eqref{alignrx4} can be written as
\begin{eqnarray*}
\text{span}(\barxH^{[43]}\barxV^{[1]}_u)=\text{span}(\barxH^{[43]}\mathbf{T}^{[1]}_1
\barxV^{[2]}) \subset \text{span} (\left[ \barxH^{[41]}\barxV^{[1]}~
\barxH^{[42]}\barxV^{[1]} \right])\\
\Rightarrow \text{span}(\barxH^{[43]}\mathbf{T}^{[1]}_1
\barxV^{[2]}) \subset \text{span} (\left[ \barxH^{[41]}~
\barxH^{[42]}\right]
\left[ \begin{array}{cc} \barxV^{[1]}& \mathbf{0}\\
\mathbf{0}& \barxV^{[1]} \end{array} \right])\\
\Rightarrow \text{span}(\underbrace{\left[\barxH^{[41]}~
\barxH^{[42]}\right]^{-1} \barxH^{[43]}\mathbf{T}^{[1]}_1
}_{\mathbf{T}^{[4]}}\barxV^{[2]}) \subset \text{span}
(\left[ \begin{array}{cc} \barxV^{[1]}& \mathbf{0}\\
\mathbf{0}& \barxV^{[1]} \end{array} \right])
\end{eqnarray*}
Note that $\mathbf{T}^{[4]}$ is a $2\mu_n \times \mu_n$ matrix and
can be written in a block matrix form:
\begin{equation*}
\mathbf{T}^{[4]}=\left[
\begin{array}{c}\mathbf{T}^{[4]}_1\\ \mathbf{T}^{[4]}_2 \end{array}\right] \label{T_block}
\end{equation*}
where each block $\mathbf{T}^{[4]}_i$ is a $\mu_n \times \mu_n$
matrix. Then, the above equation can be expressed as
\begin{equation*}
\text{span}( \left[ \begin{array}{c} \mathbf{T}^{[4]}_1 \barxV^{[2]}\\
\mathbf{T}^{[4]}_2 \barxV^{[2]}\\ \end{array} \right] )\subset
\text{span}(\left[ \begin{array}{cc} \barxV^{[1]}& \mathbf{0}\\
\mathbf{0}&\barxV^{[1]} \end{array} \right])
\end{equation*}
The above condition can be satisfied if
\begin{eqnarray}\label{cond4}
\left\{\begin{array}{ccc}
\mathbf{T}^{[4]}_1 \barxV^{[2]} & \prec & \barxV^{[1]} \\
\mathbf{T}^{[4]}_2 \barxV^{[2]} & \prec & \barxV^{[1]}  \\
\end{array}\right .
\end{eqnarray}
Therefore, we need to design $\barxV^{[1]}$ and $\barxV^{[2]}$ to
satisfy conditions \eqref{cond1}, \eqref{cond2}, \eqref{cond3},
\eqref{cond4}. Let $\mathbf{w}$ be a $3(n+1)^8 \times 1$ column
vector $\mathbf{w} = [1 \  1 \ \ldots \ 1]^T$. We need to choose
$2(n+1)^8$ column vectors for $\barxV^{[1]}$ and $2n^8$ column
vectors for $\barxV^{[2]}$. The sets of column vectors of
$\barxV^{[1]}$ and $\barxV^{[2]}$ are chosen to be equal to the sets
$\bar{V}^{[1]}$ and $\bar{V}^{[2]}$ where
\begin{equation*}
\begin{aligned} \bar{V}^{[1]} = &\{ \big(\prod_{i=1,2 j=1,\ldots,4} (\mathbf{T}_i^{[j]})^{\alpha_i^{[j]}}\big)\mathbf{w}:
\alpha_i^{[j]} \in \{1, \ldots, n+1\} \} \quad \cup
&\{\big(\prod_{i=1,2 j=1,\ldots,4}
(\mathbf{T}_i^{[j]})^{\beta_i^{[j]}}\big)\mathbf{w}: \beta_i^{[j]}
\in \{n+2, \ldots, 2n+2\} \}
\end{aligned}
\end{equation*}
\begin{equation*}
\begin{aligned} \bar{V}^{[2]} = &\{ \big(\prod_{i=1,2 j=1,\ldots,4} (\mathbf{T}_i^{[j]})^{\alpha_i^{[j]}}\big)\mathbf{w}:
\alpha_i^{[j]} \in \{1, \ldots, n\} \} \quad \cup
&\{\big(\prod_{i=1,2 j=1,\ldots,4}
(\mathbf{T}_i^{[j]})^{\beta_i^{[j]}}\big)\mathbf{w}: \beta_i^{[j]}
\in \{n+2, \ldots, 2n+1\} \}
\end{aligned}
\end{equation*}
For example, when $n=1$, the set $\bar{V}^{[2]}$ consists of two
elements, i.e., \\$\bar{V}^{[2]}= \{(\prod_{i=1,2 j=1,\ldots,4}
\mathbf{T}_i^{[j]})\mathbf{w} \quad (\prod_{i=1,2 j=1,\ldots,4}
(\mathbf{T}_i^{[j]})^3)\mathbf{w}\}$. The set $\bar{V}^{[1]}$
consists of $2(1+1)^8= 2^9$ column vectors in the form
$\{(\prod_{i=1,2 j=1,\ldots,4}
(\mathbf{T}_i^{[j]})^{\alpha_i^{[j]}})\mathbf{w} \quad (\prod_{i=1,2
j=1,\ldots,4} (\mathbf{T}_i^{[j]})^{\beta_i^{[j]}})\mathbf{w}\}$
where $\alpha_i^{[j]}$ takes values $1,2$; $\beta_i^{[j]}$ takes
values $3, 4$. Note that the above construction requires the
commutative property of multiplication of matrices
$\mathbf{T}^{[j]}_i$. Therefore, it requires $\mathbf{T}^{[j]}_i$ to
be diagonal matrices. We provide the proof to show this is true in
Appendix \ref{apdx:simo}. In order for each user to decode its
desired message by zero forcing the interference, it is required
that the desired signal vectors are linearly independent of the
interference vectors. We also show this is true in Appendix
\ref{apdx:simo}.

{\it Remark:} Note that for the $K$ user Gaussian interference
channel with  single antenna nodes\cite{Cadambe_Jafar_int} and $M
\times N$ user $X$ channel \cite{Cadambe_Jafar_X}, we need to
construct two precoding matrices $\xV$ and $\xV'$ to satisfy several
such conditions $\xV \prec \mathbf{T}_i \xV'$. Here, we use the same
precoding matrix $\barxV^{[1]}$ for Transmitter 1, 2, 3 so that we
need to design two precoding matrices $\barxV^{[1]}$ and
$\barxV^{[2]}$ to satisfy similar conditions $\barxV^{[2]} \prec
\mathbf{T}_i \barxV^{[1]}$. Therefore, we use the same method in
\cite{Cadambe_Jafar_int} and \cite{Cadambe_Jafar_X} to design
$\barxV^{[1]}$ and $\barxV^{[2]}$ here.

We present the general result for the achievable degrees of freedom
of the SIMO Gaussian interference channel in the following theorem.
\begin{theorem}\label{thm:simo}
For the $K>R+1$ user SIMO Gaussian interference channel with a
single antenna at each transmitter and $R$ antennas at each
receiver, a total of $\frac{R}{R+1}K$ degrees of freedom per
orthogonal time dimension can be achieved.
\end{theorem}
\begin{proof}
We provide the proof in Appendix \ref{apdx:simo}.
\end{proof}

Next, we present the innerbound on the degrees of freedom for the
$K$ user MIMO Gaussian interference channel in the following
theorem:
\begin{theorem}\label{thm:innerbound}
For the time-varying $K$ user MIMO Gaussian interference channel
with channel coefficients drawn from a continuous distribution and
$M$ antennas at each transmitter and $N$ antennas at each receiver,
$K\min(M,N)$ degrees of freedom can be achieved if $K \leq R$ and
$\frac{R}{R+1}\min(M,N)K$ degrees of freedom can be achieved if
$K>R$ where $R=\lfloor\frac{\max(M,N)}{\min{(M,N)}}\rfloor$, i.e.
\begin{equation*}
\eta=d_1+\cdots+d_K \geq \min{(M,N)}K~1(K \leq
R)+\frac{R}{R+1}\min(M,N)K~1(K>R)
\end{equation*}
where 1(.) is the indicator function and $d_i$ represents the
individual degrees of freedom achieved by user $i$.
\end{theorem}
\begin{proof}
When $K\leq R$, the achievable scheme is based on beamforming and
zero forcing. There is a reciprocity of such scheme discussed in
\cite{Cadambe_Jafar_X}. It is shown that the degrees of freedom is
unaffected if all transmitters and receivers are switched. For
example, the degrees of freedom of the $2$ user MISO interference
channel with 2 transmit antennas and a single receive antenna is the
same as that of the 2 user SIMO interference channel with a single
transmit antenna and 2 receive antennas. When $K>R$, the achievable
scheme is based on interference alignment. There is a reciprocity of
alignment which shows that if interference alignment is feasible on
the original channel then it is also feasible on the reciprocal
channel \cite{Gomadam_Cadambe_Jafar_dist}. Therefore, without loss
of generality, we assume that the number of transmit antennas is
less than or equal to that of receive antennas, i.e. $M \leq N$. As
a result, we need to show that $KM$ degrees of freedom can be
achieved if $K \leq R$ and $\frac{R}{R+1}MK$ degrees of freedom can
be achieved if $K>R$ where
$R=\lfloor\frac{N}{M}\rfloor$. The case when $R=1$ is solved in \cite{Cadambe_Jafar_int}. Therefore, we only consider the cases when $R>1$ here.\\
1) $K \leq R$: Each transmitter sends $M$ independent data streams
along beamforming vectors. Each receiver gets $M$ interference free
streams by zero forcing the interference from unintended
transmitters. As a result, each user can achieve $M$ degrees of
freedom for a total of $KM$ degrees of freedom.\\
2) $K>R$: When $K=R+1$, by discarding one user, we have a $R$ user
interference channel. $RM$ degrees of freedom can be achieved on
this channel using the achievable scheme described above. When
$K>R+1$, first we get $RM$ antennas receive nodes by discarding
$N-RM$ antennas at each receiver. Then, suppose we view each user
with $M$ antennas at the transmitter and $RM$ antennas at the
receiver as $M$ different users each of which has a single transmit
antenna and $R$ receive antennas. Then, instead of a $K$ user MIMO
interference channel we obtain a $KM$ user SIMO interference channel
with $R$ antennas at each receiver. By the result of Theorem
\ref{thm:simo}, $\frac{R}{R+1}KM$ degrees of freedom can be achieved
on this interference channel. Thus, we can also achieve
$\frac{R}{R+1}KM$ degrees of freedom on the $K$ user MIMO
interference channel with time-varying channel coefficients.
\end{proof}
Finally, we show that the innerbound and outerbound are tight when
the ratio $\frac{\max(M,N)}{\min(M,N)}$ is equal to an integer. We
present the result in the following corollary.
\begin{corollary}\label{thm:mimo}
For the time-varying $K$ user MIMO Gaussian interference channel
with $M$ transmit antennas and $N$ receive antennas, the total
number of degrees of freedom is equal to $K\min(M,N)$ if $K \leq R$
and $\frac{R}{R+1}\min(M,N)K$ if $K>R$ when
$R=\frac{\max(M,N)}{\min(M,N)}$ is equal to an integer, i.e.
\begin{equation*}
\eta=d_1+\cdots+d_K = \min{(M,N)}K~1(K \leq
R)+\frac{R}{R+1}\min(M,N)K~1(K>R)
\end{equation*}
\end{corollary}
\begin{proof}
The proof is obtained by directly verifying that the innerbound and
outerbound match when the ratio $R=\frac{\max(M,N)}{\min(M,N)}$ is
equal to an integer. When $K \leq R$, the innerbound and outerbound
always match which is $\min{(M,N)}K$. When $K > R$, the innerbound
and outerbound match when $\frac{R}{R+1}\min(M,N)K =
\frac{\max(M,N)}{R+1}K$ which implies that $R\min(M,N)=\max(M,N)$.
In other words, when either the number of transmit antennas is an
integer multiple of that of receive antennas or vice versa, the
total number of degrees of freedom is equal to
$\frac{R}{R+1}\min(M,N)K$.
\end{proof}

{\it Remark}: For the $K$ user MIMO Gaussian interference channel
with $M,N$ antennas at the transmitter and the receiver
respectively, if $K \leq R$ where
$R=\lfloor\frac{\max(M,N)}{\min{(M,N)}}\rfloor$ then the total
number of degrees of freedom is $\min{(M,N)}K$. This result can be
extended to the same channel with constant channel coefficients.

{\it Remark}: If $\min(M,N)=1$, then Corollary \ref{thm:mimo} shows
that the total number of degrees of freedom of the $K$ user SIMO
Gaussian interference channel with $R$ receive antennas or the $K$
user MISO Gaussian interference channel with $R$ transmit antennas
is equal to $K~1(K \leq R)+\frac{R}{R+1}K~1(K>R)$.

\section{Achievable Degrees of Freedom for the MIMO interference channel with constant channel coefficients}
Note that the converse results and the results of the achievable
degrees of freedom based on merely zero forcing in previous sections
are also applicable to the same channel with constant channel
coefficients. The results of the achievable degrees of freedom based
on interference alignment are obtained under the assumption that the
channel coefficients are time-varying. It is not known if the
results can be extended to the same channel with constant channel
coefficients. Because the construction of precoding matrices
$\barxV^{[1]}$ and $\barxV^{[2]}$ requires commutative property of
multiplication of diagonal matrices $\mathbf{T}^{[j]}_i$. But for
the MIMO scenarios, those matrices are not diagoal and commutative
property cannot be exploited. In fact, the degrees of freedom for
the interference channel with constant channel coefficients remains
an open problem for more than 2 users. One known scenario is the 3
user MIMO Gaussian interference channel with $M$ antennas at each
node. In \cite{Cadambe_Jafar_int}, it is shown that the total number
of degrees of freedom is $\frac{3}{2}M$. The achievable scheme is
based on interference alignment on signal vectors. In
\cite{Cadambe_Jafar_Shamai}, the first known example of a $K$ user
Gaussian interference channel with single antenna nodes and constant
channel coefficients are provided to achieve the outerbound on the
degrees of freedom. The achievable scheme is based on interference
alignment on signal levels rather than signal vectors. In this
section, we will provide examples where interference alignment
combined with zero forcing can achieve more degrees of freedom than
merely zero-forcing for some MIMO Gaussian interference channels
with constant channel coefficients. More general results are
provided in Appendix \ref{apdx:mimo}.

{\it Example 1}: Consider the 4 user MIMO Gaussian interference
channel with 4 antennas at each transmitter and 8 antennas at each
receiver. Note that for the 3 user MIMO interference channel with
the same antenna deployment, the total number of degrees of freedom
is 8. Also, for the 4 user case, only 8 degrees of freedom can be
achieved by merely zero forcing. However, we will show that using
interference alignment combined with zero forcing, 9 degrees of
freedom can be achieved on this interference channel without channel
extension. In other words, the 4 user MIMO interference channel with
4, 8 antennas at each transmitter and receiver respectively can
achieve more degrees of freedom than the 3 user interference channel
with the same antenna deployment. Besides, more degrees of freedom
can be achieved on this 4 user interference channel by using
interference alignment combined with zero forcing than merely zero
forcing. Next, we show that user $1,2,3$ can achieve $d_i=2, \forall
i=1,2,3$ degrees of freedom and user 4 can achieve $d_4=3$ degrees
of freedom resulting in a total of 9 degrees of freedom achieved on
this channel. Transmitter $i$ sends message $W_i$ to Receiver $i$
using $d_i$ independently encoded streams along vectors
$\mathbf{v}^{[i]}_m$, i.e.,
\begin{eqnarray*}
\mathbf{X}^{[i]}&=&\sum_{m=1}^{2}x^{[i]}_m\mathbf{v}_m^{[i]}=\mathbf{V}^{[i]}\mathbf{X}^{i},~i=1,2,3\\
\mathbf{X}^{[4]}&=&\sum_{m=1}^{3}x^{[4]}_m\mathbf{v}_m^{[4]}=\mathbf{V}^{[4]}\mathbf{X}^{4}
\end{eqnarray*}
where $\xV^{[i]}=[\xv^{[i]}_1~\xv^{[i]}_2], i=1,2,3$ and
$\xV^{[4]}=[\xv^{[4]}_1~\xv^{[4]}_2~\xv^{[4]}_3]$. The signal at
Receiver $j$ can be written as
\begin{equation*}
\mathbf{Y}^{[j]}=\sum_{i=1}^4\mathbf{H}^{[ji]}\mathbf{V}^{[i]}\mathbf{X}^{i}+\mathbf{Z}^{[j]}.
\end{equation*}
In order for each receiver to decode its message by zero forcing the
interference signals, the dimension of the space spanned by the
interference signal vectors has to be less than or equal to $8-d_i$.
Since there are $9-d_i$ interference vectors at receiver $i$, we
need to align $(9-d_i)-(8-d_i)=1$ interference signal vector at each
receiver. This can be achieved by if one interference vector lies in
the space spanned by other interference vectors at each receiver.
Mathematically, we choose the following alignments
\begin{eqnarray}
\text{span}(\mathbf{H}^{[14]}\mathbf{v}^{[4]}_1) \subset
\text{span}(\left[ \mathbf{H}^{[12]}\mathbf{V}^{[2]}~
\mathbf{H}^{[13]}\mathbf{V}^{[3]} \right]) &\Rightarrow&
\text{span}(\underbrace{[\xH^{[12]}~\xH^{[13]}]^{-1}\mathbf{H}^{[14]}}_{\mathbf{T}^{[1]}}
\mathbf{v}^{[4]}_1) \subset \text{span}(\left[ \begin{array}{cc} \xV^{[2]}& \mathbf{0}\\
\mathbf{0}&\xV^{[3]} \end{array} \right]) \notag\\
&\Rightarrow&
\text{span}( \left[ \begin{array}{c} \mathbf{T}^{[1]}_1 \xv^{[4]}_1\\
\mathbf{T}^{[1]}_2 \xv^{[4]}_1\\ \end{array} \right] )\subset
\text{span}(\left[ \begin{array}{cc} \xV^{[2]}& \mathbf{0}\\
\mathbf{0}&\xV^{[3]} \end{array} \right])\label{iacon1}
\end{eqnarray}
\begin{eqnarray}
\text{span}(\mathbf{H}^{[24]}\mathbf{v}^{[4]}_1) \subset
\text{span}(\left[ \mathbf{H}^{[21]}\mathbf{V}^{[1]}~
\mathbf{H}^{[23]}\mathbf{V}^{[3]} \right]) &\Rightarrow&
\text{span}(\underbrace{[\mathbf{H}^{[21]}~\mathbf{H}^{[23]}]^{-1}\mathbf{H}^{[24]}}_{\mathbf{T}^{[2]}}\mathbf{v}^{[4]}_1)
\subset \text{span}(\left[ \begin{array}{cc} \xV^{[1]}& \mathbf{0}\\
\mathbf{0}&\xV^{[3]} \end{array} \right])\notag\\
&\Rightarrow&
\text{span}( \left[ \begin{array}{c} \mathbf{T}^{[2]}_1 \xv^{[4]}_1\\
\mathbf{T}^{[2]}_2 \xv^{[4]}_1\\ \end{array} \right] )\subset
\text{span}(\left[ \begin{array}{cc} \xV^{[1]}& \mathbf{0}\\
\mathbf{0}&\xV^{[3]} \end{array} \right])\label{iacon2}
\end{eqnarray}
\begin{eqnarray}
\text{span}(\mathbf{H}^{[32]}\mathbf{v}^{[2]}_1) \subset
\text{span}(\left[\mathbf{H}^{[31]}\mathbf{V}^{[1]}~
\mathbf{H}^{[34]}\mathbf{V}^{[4]} \right]) &\Rightarrow&
\text{span}(\underbrace{[\mathbf{H}^{[31]}~\mathbf{H}^{[34]}]^{-1}\mathbf{H}^{[32]}}_{\mathbf{T}^{[3]}}\mathbf{v}^{[2]}_1)
\subset \text{span}(\left[ \begin{array}{cc} \xV^{[1]}& \mathbf{0}\\
\mathbf{0}&\xV^{[4]} \end{array} \right]))\notag\\
&\Rightarrow&
\text{span}( \left[ \begin{array}{c} \mathbf{T}^{[3]}_1 \xv^{[2]}_1\\
\mathbf{T}^{[3]}_2 \xv^{[2]}_1\\ \end{array} \right] )\subset
\text{span}(\left[ \begin{array}{cc} \xV^{[1]}& \mathbf{0}\\
\mathbf{0}&\xV^{[4]} \end{array} \right]) \label{iacon3}\\
\text{span}(\mathbf{H}^{[41]}\mathbf{v}^{[1]}_1) \subset
\text{span}(\left[ \mathbf{H}^{[42]}\mathbf{V}^{[2]}
~\mathbf{H}^{[43]}\mathbf{V}^{[3]} \right]) &\Rightarrow&
\text{span}(\underbrace{[\mathbf{H}^{[42]}~\mathbf{H}^{[43]}]^{-1}\mathbf{H}^{[41]}}_{\mathbf{T}^{[4]}}\mathbf{v}^{[1]}_1)
\subset \text{span}(\left[ \begin{array}{cc} \xV^{[2]}& \mathbf{0}\\
\mathbf{0}&\xV^{[3]} \end{array}
\right]) \notag\\
&\Rightarrow&
\text{span}( \left[ \begin{array}{c} \mathbf{T}^{[4]}_1 \xv^{[1]}_1\\
\mathbf{T}^{[4]}_2 \xv^{[1]}_1\\ \end{array} \right] )\subset
\text{span}(\left[ \begin{array}{cc} \xV^{[2]}& \mathbf{0}\\
\mathbf{0}&\xV^{[3]} \end{array} \right])\label{iacon4}
\end{eqnarray}
where $\mathbf{T}^{[i]}$ is an $8 \times 4$ matrix which can be
written in a block matrix form:
\begin{equation}
\mathbf{T}^{[i]}=\left[\begin{array}{c}\mathbf{T}^{[i]}_{1}\\
\mathbf{T}^{[i]}_{2}\end{array}\right]~~i=1,2,3,4
\end{equation}
where $\mathbf{T}^{[i]}_{1}$ and $\mathbf{T}^{[i]}_{2}$ are $4
\times 4$ matrices. To satisfy the conditions \eqref{iacon1},
\eqref{iacon2}, \eqref{iacon3}, \eqref{iacon4}, we let
\begin{eqnarray*}
\mathbf{T}^{[1]}_1 \xv^{[4]}_1 = \xv^{[2]}_1 &~&
\text{span}(\mathbf{T}^{[1]}_2 \xv^{[4]}_1) =\text{span}
(\xv^{[3]}_1)\label{mimo_eigen2}\\
\mathbf{T}^{[2]}_1 \xv^{[4]}_1 = \xv^{[1]}_1 &~&
\text{span}(\mathbf{T}^{[2]}_2 \xv^{[4]}_1) = \text{span}
(\xv^{[3]}_1)\label{mimo_eigen3}\\
\mathbf{T}^{[3]}_1 \xv^{[2]}_1 =\xv^{[1]}_2 &~&
\mathbf{T}^{[3]}_2 \xv^{[2]}_1 = \xv^{[4]}_2\\
\mathbf{T}^{[4]}_1 \xv^{[1]}_1 = \xv^{[2]}_2 &~& \mathbf{T}^{[4]}_2
\xv^{[1]}_1 = \xv^{[3]}_2
\end{eqnarray*}
Notice once $\xv^{[4]}_1$ is chosen, all other vectors can be solved
from the above equations. To solve $\xv^{[4]}_1$, we have
\begin{eqnarray*}
\text{span}(\mathbf{T}^{[1]}_1 \xv^{[4]}_1)&=& \text{span}(\mathbf{T}^{[2]}_2 \xv^{[4]}_1)\\
\Rightarrow
\text{span}((\mathbf{T}^{[2]}_{2})^{-1}\mathbf{T}^{[1]}_{2}\mathbf{v}^{[4]}_{1})&=& \text{span} (\mathbf{v}^{[4]}_{1})\\
\Rightarrow \mathbf{v}^{[4]}_{1}&=&\mathbf{e},
\end{eqnarray*}
where $\mathbf{e}$ is an eigenvector of matrix
$(\mathbf{T}^{[2]}_2)^{-1}\mathbf{T}^{[1]}_2$. Note that the above
construction only specifies $\mathbf{V}^{[i]}, \forall i=1,2,3$ and
$\xv^{[4]}_1, \xv^{[4]}_2$. The remaining $\mathbf{v}^{[4]}_3$ can
be picked randomly according to a continuous distribution so that
all columns of $\mathbf{V}^{[i]}$ are linearly independent.

Through interference alignment, we ensure that the interference
vectors span a small enough signal space. We need to verify that the
desired signal vectors, i.e., $\mathbf{H}^{[ii]}\mathbf{V}^{[i]}$
are linearly independent of interference vectors so that each
receiver can decode its message using zero forcing. Notice that the
direct channel matrices $\mathbf{H}^{[ii]}, i=1,2,3,4$ do not appear
in the interference alignment equations, $\mathbf{V}^{[i]}$
undergoes an independent linear transformation by multiplying
$\mathbf{H}^{[ii]}$. Therefore, at each receiver the desired signal
vectors are linearly independent of the interference signal vectors
with probability one. As a result, user $i$ can achieve $d_i$
degrees of freedom and a total of 9 degrees of freedom can be
achieved.

{\it Example 2}: Consider the 4 user MIMO Gaussian interference
channel with 2 antennas at each transmitter and 4 antennas at each
receiver. We show that 9 degrees of freedom can be achieved on the
2-symbol extension of the original channel and hence 4$\frac{1}{2}$
degrees of freedom per channel use can be achieved. Since only 4
degrees of freedom can be achieved using merely zero forcing,
$\frac{1}{2}$ more degrees of freedom is achieved using interference
alignment scheme. Note that although we have equivalently a 4 user
interference channel with $4 \times 8$ channel on the 2-symbol
extension channel, we cannot use the same achievable scheme used in
Example 1 due to the block diagonal structure of the extension
channel matrix. Consider 2-symbol extension of the channel. The
channel input-output relationship is
\begin{equation*}
\barxY^{[j]}=
\sum_{i=1}^{4}\barxH^{[ji]}\barxX^{[i]}+\barxZ^{[j]}~~\forall
j=1,2,3,4
\end{equation*}
where the overbar notation represents the 2-symbol extensions so
that
\begin{equation*}
\mathbf{\bar{X}}\triangleq\left[\begin{array}{c}\mathbf{X}(2t)\\
\mathbf{X}(2t+1)\end{array}\right]\quad
\mathbf{\bar{Z}}\triangleq\left[\begin{array}{c}\mathbf{Z}(2t)\\
\mathbf{Z}(2t+1)\end{array}\right]
\end{equation*}
where $\mathbf{X}$ and $\mathbf{Z}$ are $2 \times 1$ and $4 \times
1$ vectors respectively, and
\begin{equation*}
\barxH \triangleq \left[\begin{array}{cc}\xH & \mathbf{0}\\
\mathbf{0} & \xH \end{array}\right].
\end{equation*}
where $\xH$ is the $4 \times 2$ channel matrix. We assign
$d_1=d_2=d_3=2$ and $d_4=3$ degrees of freedom to message
$W_1,W_2,W_3,W_4$ respectively for a total 9 degrees of freedom over
the 2-symbol extension channel. Transmitter $i$ sends message $W_i$
in the form of $d_i$ independently encoded streams along the
direction vectors $\barxv^{[i]}_1,\ldots,\barxv^{[i]}_{d_i}$, each
of dimension $4 \times 1$, so that we have:
\begin{equation*}
\barxX^{[i]}=\sum_{m=1}^{d_i}\barxv^{[i]}_{m}x^{[i]}_m=\barxV^{[i]}\xX^{[i]}\quad
i=1,2,3,4
\end{equation*}
where $\barxV^{[i]}$ and $\xX^{[i]}$ are $4 \times d_i$ and $d_i
\times 1$ matrices respectively. In order to get $d_i$ interference
free dimension at Receiver $i$, we need to align 1 interference
vector at each receiver. This can be achieved if one interference
vector lies in the space spanned by other interference vectors at
each receiver. Mathematically, we choose the following alignments:
\begin{eqnarray}
\text{span}(\barxH^{[12]}\barxv^{[2]}_1) \subset \text{span}(\left[
\barxH^{[13]}\barxV^{[3]}~ \barxH^{[14]}\barxV^{[4]}
\right])\Rightarrow
\text{span}(\underbrace{[\barxH^{[13]}~\barxH^{[14]}]^{-1}\barxH^{[12]}}_{\mathbf{T}^{[1]}}
\barxv^{[2]}_1) \subset \text{span}(\left[ \begin{array}{cc}
\barxV^{[3]}& \mathbf{0}\\  \mathbf{0}&\barxV^{[4]} \end{array} \right])\\
\text{span}(\barxH^{[23]}\barxv^{[3]}_1) \subset\text{span}(\left[
\barxH^{[21]}\barxV^{[1]}~ \barxH^{[24]}\barxV^{[4]} \right])
\Rightarrow
\text{span}(\underbrace{[\barxH^{[21]}~\barxH^{[24]}]^{-1}\barxH^{[23]}}_{\mathbf{T}^{[2]}}\barxv^{[3]}_1)
\subset \text{span}(\left[ \begin{array}{cc} \barxV^{[1]}& \mathbf{0}\\
\mathbf{0}& \barxV^{[4]} \end{array} \right])\\
\text{span}(\barxH^{[34]}\barxv^{[4]}_1) \subset
\text{span}(\left[\barxH^{[31]}\barxV^{[1]}~ \barxH^{[32]}
\barxV^{[2]} \right])\Rightarrow
\text{span}(\underbrace{[\barxH^{[31]}~\barxH^{[32]}]^{-1}\barxH^{[34]}}_{\mathbf{T}^{[3]}}\barxv^{[4]}_1
\subset \text{span}(\left[ \begin{array}{cc} \barxV^{[1]}& \mathbf{0}\\
\mathbf{0}& \barxV^{[2]} \end{array} \right]))\\
\text{span}(\barxH^{[41]}\barxv^{[1]}_1) \subset
\text{span}(\left[\barxH^{[42]} \barxV^{[2]}
~\barxH^{[43]}\barxV^{[3]} \right]) \Rightarrow
\text{span}(\underbrace{[\barxH^{[42]}~\barxH^{[43]}]^{-1}\barxH^{[41]}}_{\mathbf{T}^{[4]}}\barxv^{[1]}_1)
\subset \text{span}(\left[ \begin{array}{cc} \barxV^{[2]}& \mathbf{0}\\
\mathbf{0}& \barxV^{[3]} \end{array} \right])
\end{eqnarray}
where $\mathbf{T}^{[i]}$ is the $8 \times 4$ matrix which can be
written in a block matrix form:
\begin{equation}
\mathbf{T}^{[i]}=\left[\begin{array}{c}\mathbf{T}^{[i]}_{1}\\
\mathbf{T}^{[i]}_{2}\end{array}\right]~~i=1,2,3,4
\end{equation}
The above equations can be satisfied if
\begin{eqnarray}
\mathbf{T}^{[1]}\barxv^{[2]}_1 =
\left[\begin{array}{c}\barxv^{[3]}_1\\\barxv^{[4]}_1
\end{array}\right]~ \mathbf{T}^{[2]}\barxv^{[3]}_1 =
\left[\begin{array}{c}\barxv^{[1]}_1\\\barxv^{[4]}_2
\end{array}\right]~ \mathbf{T}^{[3]}\barxv^{[4]}_1 =
\left[\begin{array}{c}\barxv^{[1]}_2\\\barxv^{[2]}_2
\end{array}\right]~ \mathbf{T}^{[4]}\barxv^{[1]}_1 =
\left[\begin{array}{c}\barxv^{[2]}_3\\\barxv^{[3]}_2
\end{array}\right]
\end{eqnarray}
Notice that once we pick $\barxv^{[2]}_1$, all other vectors can be
solved from above equations. $\barxv^{[2]}_1$ can be chosen randomly
according to a continuous distribution so that all vectors are
linearly independent with probability one. Also, since all the
vectors are chosen independently of the direct channel matrices
$\barxH^{[ii]}$ and all entries of $\barxV^{[i]}$ are not equal to
zero almost surely, the desired signal vectors are linearly
independent of the interference vectors at each receiver. As a
result, Receiver $i$ can decode its message by zero forcing the
interference to achieve $d_i$ degrees of freedom for a total of 9
degrees of freedom over the 2-symbol extension channel. Therefore,
$4\frac{1}{2}$ degrees of freedom per channel use can be achieved on
the original channel.

\section{conclusion}
We investigate the degrees of freedom for the $K$ user MIMO Gaussian
interference channel with $M,N$ antennas at each transmitter and
receiver, respectively. The motivation of this work is the potential
benefits of interference alignment scheme shown recently to achieve
the capacity of certain wireless networks within $o(\log(SNR))$. In
this work, interference alignment scheme is also found to be optimal
in achieving the degrees of freedom of the $K$ user $M \times N$
MIMO Gaussian interference channel if the ratio
$\frac{\max(M,N)}{\min(M,N)}$ is equal to an integer with
time-varying channel coefficients drawn from a continuous
distribution. We also explore the achievable degrees of freedom for
the MIMO interference channel with constant channel coefficients
using interference alignment combined with zero forcing. We provide
some examples where using interference alignment can achieve more
degrees of freedom than merely zero forcing.

\appendices
\section{Proof of Theorem \ref{thm:simo}}\label{apdx:simo}
\begin{proof}
Let $\Gamma =KR(K-R-1)$. We will develop a coding scheme based on
interference alignment to achieve a total of
$(R+1)R(n+1)^{\Gamma}+(K-R-1)Rn^{\Gamma}$ degrees of freedom over a
$\mu_n=(R+1)(n+1)^\Gamma$ symbol extension of the original channel.
Hence, a total of
$\frac{(R+1)R(n+1)^{\Gamma}+(K-R-1)Rn^{\Gamma}}{(R+1)(n+1)^\Gamma}$
degrees of freedom per orthogonal dimension can be achieved for any
arbitrary $n \in \mathbb{N}$. Taking supremum over all $n$ proves
the total number of degrees of freedom is equal to $\frac{RK}{R+1}$
as desired. Specifically, over the extended channel, user
$i=1,2,\cdots,R+1$ achieves $R(n+1)^{\Gamma}$ degrees of freedom and
other user $i=R+2,R+3,\cdots,K$ achieves $Rn^{\Gamma}$ degrees of
freedom. As a result, user $i=1,2,\cdots,R+1$ achieves
$\frac{R(n+1)^{\Gamma}}{(R+1)(n+1)^\Gamma}$ degrees of freedom and
user $i=R+2,R+3,\cdots,K$ achieves
$\frac{Rn^{\Gamma}}{(R+1)(n+1)^\Gamma}$ degrees of freedom per
channel use, i.e.
\begin{equation}
d_i = \frac{R(n+1)^{\Gamma}}{(R+1)(n+1)^\Gamma}~~~~ i=1,2,\cdots,R+1
\quad d_i = \frac{Rn^{\Gamma}}{(R+1)(n+1)^\Gamma}~~~~
i=R+2,R+3,\cdots,K
\end{equation}
This implies that
\begin{equation}
d_1 + d_2 + \cdots+ d_K \geq \sup_n
\frac{(R+1)R(n+1)^{\Gamma}+(K-R-1)Rn^{\Gamma}}{(R+1)(n+1)^\Gamma} =
 \frac{KR}{R+1}
\end{equation}
In the extended channel, the signal vector at the $k^{th}$ user's
receiver can be expressed as
\begin{equation*}
\bar{\mathbf{Y}}^{[k]}(t) =
\sum_{j=1}^{K}\bar{\mathbf{H}}^{[kj]}(t)\bar{\mathbf{X}}^{[j]}(t)+\bar{\mathbf{Z}}^{[k]}(t)
\end{equation*}
where $\bar{\mathbf{X}}^{[j]}(t)$ is a $\mu_n \times 1$ column
vector representing the $\mu_n$ symbol extension of the transmitted
symbol $x^{[j]}(t)$, i.e.
\begin{equation*}
\bar{\mathbf{X}}^{[j]}(t) \triangleq
\left[\begin{array}{c}x^{[j]}(\mu_n(t-1)+1)\\x^{[j]}(\mu_n(t-1)+2)\\\vdots\\x^{[j]}(\mu_nt)\end{array}\right]
\end{equation*}
Similarly, $\bar{\mathbf{Y}}(t)$ and $\bar{\mathbf{Z}}(t)$ represent
$\mu_n$ symbol extensions of the $\mathbf{Y}(t)$ and $\mathbf{Z}(t)$
respectively. $\bar{\mathbf{H}}^{[kj]}(t)$ is a $R\mu_n \times
\mu_n$ matrix representing the $\mu_n$ symbol extension of the
channel, i.e.
\begin{align}
  \bar{\mathbf{H}}^{[kj]}(t) =  \left[ \begin{array}{cccc}  \mathbf{h}^{[kj]}(\mu_n(t-1)+1) & \mathbf{0} & \ldots & \mathbf{0}\\
     \mathbf{0} & \mathbf{h}^{[kj]}(\mu_n(t-1)+2) & \ldots & \mathbf{0}\\
    \vdots & \vdots & \ddots & \vdots\\
     \mathbf{0} & \mathbf{0}& \cdots  & \mathbf{h}^{[kj]}(\mu_nt)
    \end{array}\right]
\end{align}
where $\mathbf{h}^{[kj]}$ is the $R \times 1$ channel vector.
Message $W_j$ ($j=1,2,\cdots,R+1$) is encoded at Transmitter $j$
into $R(n+1)^{\Gamma}$ independent streams $x^{[j]}_m(t)$,
$m=1,2,\ldots,R(n+1)^{\Gamma}$ along the same set of vectors
$\barxv^{[1]}_m(t)$ so that $\bar{\mathbf{X}}^{[j]}(t)$ is
\begin{equation*}
\barxX^{[j]}(t) = \sum_{m=1}^{R(n+1)^{\Gamma}} x^{[j]}_m(t)
\barxv_m^{[1]}(t) = \barxV^{[1]}(t) \xX^{[j]}(t)
\end{equation*}
where  $\mathbf{X}^{[j]}(t)$ is a $R(n+1)^{\Gamma}\times 1$ column
vector and $\bar{\mathbf{V}}^{[1]}(t)$ is a $ (R+1)(n+1)^\Gamma
\times R(n+1)^{\Gamma} $ dimensional matrix. Similarly, $W_j$
($j=R+2,\cdots,K$) is encoded at Transmitter $j$ into $Rn^{\Gamma}$
independent streams $x^{[j]}_m(t)$, $m=1,2,\ldots,Rn^{\Gamma}$ along
the same set of vectors $\barxv^{[2]}_m(t)$ so that
\begin{equation*}
\barxX^{[j]}(t) = \sum_{m=1}^{Rn^{\Gamma}} x^{[j]}_m(t)
\barxv_m^{[2]}(t) = \barxV^{[2]}(t) \xX^{[j]}(t)
\end{equation*}
The received signal at the $k^{th}$ receiver can
then be written as
\begin{equation*}
\barxY^{[k]}(t) = \sum_{j=1}^{R+1}\barxH^{[kj]}(t)
\barxV^{[1]}(t)\xX^{[j]}(t) +
\sum_{j=R+2}^{K}\barxH^{[kj]}(t)\barxV^{[2]}(t)\xX^{[j]}(t) +
\bar{\mathbf{Z}}^{[k]}(t)
\end{equation*}
We wish to design the direction vectors $\bar{\mathbf{V}}^{[1]}$ and
$\bar{\mathbf{V}}^{[2]}$ so that signal spaces are aligned at
receivers where they constitute interference while they are
separable at receivers where they are desired. As a result, each
receiver can decode its desired signal by zero forcing the
interference signals.

First consider Receiver $k$, $\forall k=1,2,\cdots,R+1$. Every
receiver needs a $R(n+1)^{\Gamma}$ interference free dimension out
of the $R(R+1)(n+1)^{\Gamma}$ dimensional signal space. Thus, the
dimension of the signal space spanned by the interference signal
vectors cannot be more than $R^2(n+1)^{\Gamma}$. Notice that all the
interference vectors from Transmitter
$1,2,\cdots,k-1,k+1,\cdots,R+1$ span a $R^2(n+1)^{\Gamma}$
dimensional subspace in the $R(R+1)(n+1)^{\Gamma}$ dimensional
signal space. Hence, we can align the interference signal vectors
from Transmitter $j$, $\forall j=R+2,R+3,\cdots,K$ within this
$R^2(n+1)^{\Gamma}$ dimensional subspace. Mathematically, we have
\begin{equation*}
\text{span}(\barxH^{[kj]}\barxV^{[2]}) \subset \text{span}
(\left[\bar{\mathbf{H}}^{[k1]}\barxV^{[1]}
~\bar{\mathbf{H}}^{[k2]}\barxV^{[1]} \cdots
\bar{\mathbf{H}}^{[k(k-1)]}\barxV^{[1]}~\bar{\mathbf{H}}^{[k(k+1)]}\barxV^{[1]}
\cdots \bar{\mathbf{H}}^{[k(R+1)]}\barxV^{[1]}\right])
\end{equation*}
where $\text{span} (\mathbf{A})$ represents the space spanned by the
columns of matrix $\mathbf{A}$. The above equation can be expressed
equivalently as
\begin{equation*}
\text{span}(\barxH^{[kj]}\barxV^{[2]}) \subset \text{span}
(\left[\bar{\mathbf{H}}^{[k1]}~\bar{\mathbf{H}}^{[k2]} \cdots
\bar{\mathbf{H}}^{[k(k-1)]}~\bar{\mathbf{H}}^{[k(k+1)]} \cdots
\bar{\mathbf{H}}^{[k(R+1)]}\right] \begin{tiny} \left[\begin{array}{ccccccc}\barxV^{[1]}&\mathbf{0}& \cdots &\mathbf{0}& \mathbf{0}&\cdots&\mathbf{0}\\ \mathbf{0}&\barxV^{[1]}&\cdots&\mathbf{0}&\mathbf{0}&\cdots&\mathbf{0}\\
\vdots& \vdots & \ddots &\vdots &\vdots &\ddots &\vdots\\
\mathbf{0}&\mathbf{0}& \cdots &\barxV^{[1]}& \cdots & \cdots & \mathbf{0}\\
\mathbf{0}&\mathbf{0}& \cdots & \cdots & \barxV^{[1]}& \cdots &\mathbf{0}\\
\vdots& \vdots & \ddots &\vdots &\vdots &\ddots &\vdots\\
\mathbf{0}& \mathbf{0} & \cdots &\mathbf{0} &\mathbf{0} &\cdots &
\barxV^{[1]}
\\ \end{array} \right]) \end{tiny}
\end{equation*}
Notice that $[\bar{\mathbf{H}}^{[k1]}~\bar{\mathbf{H}}^{[k2]} \cdots
\bar{\mathbf{H}}^{[k(k-1)]}~\bar{\mathbf{H}}^{[k(k+1)]} \cdots
\bar{\mathbf{H}}^{[k(R+1)]}]$ is a $R\mu_n \times R\mu_n$ square
matrix with full rank almost surely. Thus, the above equation can be
expressed equivalently as
\begin{eqnarray}\label{eqnreq1}
\text{span}(\underbrace{\left[\bar{\mathbf{H}}^{[k1]}~\bar{\mathbf{H}}^{[k2]}
\cdots \bar{\mathbf{H}}^{[k(k-1)]}~\bar{\mathbf{H}}^{[k(k+1)]}
\cdots \bar{\mathbf{H}}^{[k(R+1)]}\right]^{-1}
\barxH^{[kj]}}_{\mathbf{T}^{[kj]}}\barxV^{[2]}) \subset \notag \\
\text{span}
(\begin{tiny} \left[\begin{array}{ccccccc}\barxV^{[1]}&\mathbf{0}& \cdots &\mathbf{0}& \mathbf{0}&\cdots&\mathbf{0}\\ \mathbf{0}&\barxV^{[1]}&\cdots&\mathbf{0}&\mathbf{0}&\cdots&\mathbf{0}\\
\vdots& \vdots & \ddots &\vdots &\vdots &\ddots &\vdots\\
\mathbf{0}&\mathbf{0}& \cdots &\barxV^{[1]}& \cdots & \cdots & \mathbf{0}\\
\mathbf{0}&\mathbf{0}& \cdots & \cdots & \barxV^{[1]}& \cdots &\mathbf{0}\\
\vdots& \vdots & \ddots &\vdots &\vdots &\ddots &\vdots\\
\mathbf{0}& \mathbf{0} & \cdots &\mathbf{0} &\mathbf{0} &\cdots &
\barxV^{[1]}
\\ \end{array} \right]) \end{tiny}
\end{eqnarray}
Note that $\mathbf{T}^{[kj]}$ is a $R\mu_n \times \mu_n$ matrix and
can be written in a block matrix form:
\begin{equation}
\mathbf{T}^{[kj]}=\left[
\begin{array}{c}\mathbf{T}^{[kj]}_1\\\mathbf{T}^{[kj]}_2\\\vdots\\
\mathbf{T}^{[kj]}_R\end{array}\right] \label{T_block}
\end{equation}
where each block $\mathbf{T}^{[kj]}_i$ is a $\mu_n \times \mu_n$
matrix. Then, \eqref{eqnreq1} can be expressed equivalently as
\begin{equation*}
\text{span}(
\left[ \begin{array}{c} \mathbf{T}^{[kj]}_1 \barxV^{[2]}\\ \mathbf{T}^{[kj]}_2 \barxV^{[2]}\\ \vdots\\
\mathbf{T}^{[kj]}_R \barxV^{[2]} \end{array} \right] )\subset
\text{span}(
\left[ \begin{array}{cccc}\barxV^{[1]}&\mathbf{0}& \cdots &\mathbf{0}\\
\mathbf{0}&\barxV^{[1]}&\cdots&\mathbf{0}\\ \vdots& \vdots & \ddots
&\vdots\\ \mathbf{0}& \mathbf{0} & \cdots & \barxV^{[1]}
\end{array}\right])
\end{equation*}
The above condition can be satisfied if
\begin{equation}\label{c1}
\mathbf{T}^{[kj]}_i\bar{\mathbf{V}}^{[2]} \prec
\bar{\mathbf{V}}^{[1]}~\forall k=1,\ldots,R+1
~j=R+2,\ldots,K~i=1,\ldots,R
\end{equation}
where $\mathbf{P}  \prec  \mathbf{Q}$ means that the set of column
vectors of matrix $\mathbf{P}$ is a subset of the set of column
vectors of matrix $\mathbf{Q}$.

Then consider Receiver $k$, $\forall k=R+2,R+3,\cdots,K$. To get a
$Rn^{\Gamma}$ interference free dimension signal space, the
dimension of the signal space spanned by the interference vectors
cannot be more than $R(R+1)(n+1)^{\Gamma}-Rn^{\Gamma}$ at each
receiver. This can be achieved if all interference vectors from
Transmitter $j$, $\forall j=R+2,\cdots,k-1,k+1,\cdots,K$ and
$Rn^{\Gamma}$ interference vectors from Transmitter $R+1$ are
aligned within the signal space spanned by interference vectors from
transmitter $1,2,\cdots,R$. We first consider aligning the
interference from Transmitter $R+2,\cdots,k-1,k+1,\cdots,K$.
Mathematically, we choose the following alignments:
\begin{eqnarray*}
\text{span}(\barxH^{[kj]}\barxV^{[2]}) &\subset& \text{span}
(\left[\bar{\mathbf{H}}^{[k1]}\barxV^{[1]}
~\bar{\mathbf{H}}^{[k2]}\barxV^{[1]} \cdots
\bar{\mathbf{H}}^{[kR]}\barxV^{[1]}\right])\\
\Rightarrow \text{span}(\barxH^{[kj]}\barxV^{[2]}) &\subset&
\text{span} (\left[\bar{\mathbf{H}}^{[k1]} ~\bar{\mathbf{H}}^{[k2]}~
\cdots~ \bar{\mathbf{H}}^{[kR]}\right]
\left[ \begin{array}{cccc}\barxV^{[1]}&\mathbf{0}& \cdots &\mathbf{0}\\
\mathbf{0}&\barxV^{[1]}&\cdots&\mathbf{0}\\ \vdots& \vdots & \ddots
&\vdots\\ \mathbf{0}& \mathbf{0} & \cdots & \barxV^{[1]}
\end{array}\right])
\end{eqnarray*}
Notice that $[\bar{\mathbf{H}}^{[k1]}
~\bar{\mathbf{H}}^{[k2]}~\cdots~ \bar{\mathbf{H}}^{[kR]}]$ is a
$R\mu_n \times R\mu_n$ square matrix with full rank almost surely.
Thus, the above equation can be expressed equivalently as
\begin{eqnarray}\label{eqnreq2}
\text{span}(\underbrace{\left[\bar{\mathbf{H}}^{[k1]}
~\bar{\mathbf{H}}^{[k2]} \cdots \bar{\mathbf{H}}^{[kR]}\right]^{-1}
\barxH^{[kj]}}_{\mathbf{T}^{[kj]}}\barxV^{[2]}) \subset  \text{span}
(\left[\begin{array}{cccc}\barxV^{[1]}&\mathbf{0}& \cdots &\mathbf{0}\\
\mathbf{0}&\barxV^{[1]}&\cdots&\mathbf{0}\\ \vdots& \vdots & \ddots
&\vdots\\ \mathbf{0}& \mathbf{0} & \cdots & \barxV^{[1]}
\end{array} \right])
\end{eqnarray}
Note that $\mathbf{T}^{[kj]}$ is a $R\mu_n \times \mu_n$ matrix and
can be written in a block matrix form:
\begin{equation*}
\mathbf{T}^{[kj]}=\left[
\begin{array}{c}\mathbf{T}^{[kj]}_1\\\mathbf{T}^{[kj]}_2\\\vdots\\
\mathbf{T}^{[kj]}_R\end{array}\right] \label{T_block}
\end{equation*}
where each block $\mathbf{T}^{[kj]}_i$ is a $\mu_n \times \mu_n$
matrix. Then, \eqref{eqnreq2} can be expressed as
\begin{equation*}
\text{span}(
\left[ \begin{array}{c} \mathbf{T}^{[kj]}_1 \barxV^{[2]}\\ \mathbf{T}^{[kj]}_2 \barxV^{[2]}\\ \vdots\\
\mathbf{T}^{[kj]}_R \barxV^{[2]} \end{array} \right] )\subset
\text{span}(
\left[ \begin{array}{cccc}\barxV^{[1]}&\mathbf{0}& \cdots &\mathbf{0}\\
\mathbf{0}&\barxV^{[1]}&\cdots&\mathbf{0}\\ \vdots& \vdots & \ddots
&\vdots\\ \mathbf{0}& \mathbf{0} & \cdots & \barxV^{[1]}
\end{array}\right])
\end{equation*}
The above condition can be satisfied if
\begin{equation}
\mathbf{T}^{[kj]}_i\bar{\mathbf{V}}^{[2]} \prec
\bar{\mathbf{V}}^{[1]}
~k=R+2,R+3,\cdots,K~~j=R+2,\cdots,k-1,k+1,\cdots,K~i=1,\cdots,R
\label{c2}
\end{equation}
Now consider aligning $Rn^{\Gamma}$ interference vectors from
Transmitter $R+1$ at Receiver $k$, $\forall k=R+2,R+3,\cdots,K$.
This can be achieved if the space spanned by $Rn^{\Gamma}$ columns
of $\bar{\mathbf{H}}^{[k(R+1)]}\barxV^{[1]}$ is aligned within the
range of
$\left[\bar{\mathbf{H}}^{[k1]}\barxV^{[1]}~\cdots~\barxH^{[kR]}\barxV^{[1]}\right]$.
Since $\barxV^{[1]}$ is a $\mu_n \times R(n+1)^{\Gamma}$ matrix, we
can write it as $\barxV^{[1]}=[\barxV^{[1]}_u~
\barxV^{[1]}_{\epsilon_n}]$ where $\barxV^{[1]}_u$ and
$\barxV^{[1]}_{\epsilon_n}$ are $\mu_n \times Rn^{\Gamma}$ and
$\mu_n \times (R(n+1)^{\Gamma}-Rn^{\Gamma})$ matrices, respectively.
We assume the space spanned by the columns of
$\bar{\mathbf{H}}^{[k(R+1)]}\barxV^{[1]}_u$ is aligned within the
space spanned by the interference from Transmitter 1, 2, \ldots,
$R$. From equation \eqref{c1}, we have
\begin{equation*}
\mathbf{T}^{[1(R+2)]}_1 \bar{\mathbf{V}}^{[2]} \prec
\bar{\mathbf{V}}^{[1]}
\end{equation*}
This implies that $Rn^{\Gamma}$ columns of $\barxV^{[1]}$ are equal
to the columns of $\mathbf{T}^{[1(R+2)]}_R \bar{\mathbf{V}}^{[2]}$.
Without loss of generality, we assume that
$\barxV^{[1]}_u=\mathbf{T}^{[1(R+2)]}_1 \bar{\mathbf{V}}^{[2]}$.
Thus, to satisfy the interference alignment requirement, we choose
the following alignments:
\begin{eqnarray*}
\text{span}(\barxH^{[k(R+1)]}\bar{\mathbf{V}}^{[1]}_u)=\text{span}(\barxH^{[k(R+1)]}\mathbf{T}^{[1(R+2)]}_1
\bar{\mathbf{V}}^{[2]}) \subset \text{span}
(\left[\bar{\mathbf{H}}^{[k1]}\barxV^{[1]}
~\bar{\mathbf{H}}^{[k2]}\barxV^{[1]} \cdots
\bar{\mathbf{H}}^{[kR]}\barxV^{[1]}\right])\\
\Rightarrow \text{span}(\barxH^{[k(R+1)]}\mathbf{T}^{[1(R+2)]}_1
\bar{\mathbf{V}}^{[2]}) \subset \text{span}
(\left[\bar{\mathbf{H}}^{[k1]}~\bar{\mathbf{H}}^{[k2]}~\cdots~
\bar{\mathbf{H}}^{[kR]}\right]
\left[ \begin{array}{cccc}\barxV^{[1]}&\mathbf{0}& \cdots &\mathbf{0}\\
\mathbf{0}&\barxV^{[1]}&\cdots&\mathbf{0}\\ \vdots& \vdots & \ddots
&\vdots\\ \mathbf{0}& \mathbf{0} & \cdots & \barxV^{[1]}
\end{array}\right])\\
\Rightarrow \text{span}(\underbrace{\left[\bar{\mathbf{H}}^{[k1]}
~\bar{\mathbf{H}}^{[k2]}~\cdots~ \bar{\mathbf{H}}^{[kR]}\right]^{-1}
\barxH^{[k(R+1)]}\mathbf{T}^{[1(R+2)]}_1}_{\mathbf{T}^{[k(R+1)]}}\barxV^{[2]})
\subset \text{span}
(\left[\begin{array}{cccc}\barxV^{[1]}&\mathbf{0}& \cdots &\mathbf{0}\\
\mathbf{0}&\barxV^{[1]}&\cdots&\mathbf{0}\\ \vdots& \vdots & \ddots
&\vdots\\ \mathbf{0}& \mathbf{0} & \cdots & \barxV^{[1]}
\end{array} \right])
\end{eqnarray*}
Note that $\mathbf{T}^{[k(R+1)]}$ is a $R\mu_n \times \mu_n$ matrix
and can be written in a block matrix form:
\begin{equation*}
\mathbf{T}^{[k(R+1)]}=\left[
\begin{array}{c}\mathbf{T}^{[k(R+1)]}_1\\\mathbf{T}^{[k(R+1)]}_2\\\vdots\\
\mathbf{T}^{[k(R+1)]}_R\end{array}\right] \label{T_block}
\end{equation*}
where each block $\mathbf{T}^{[k(R+1)]}_i$ is a $\mu_n \times \mu_n$
matrix. Then, the above equation can be expressed as
\begin{equation*}
\text{span}(
\left[ \begin{array}{c} \mathbf{T}^{[k(R+1)]}_1 \barxV^{[2]}\\ \mathbf{T}^{[k(R+1)]}_2 \barxV^{[2]}\\ \vdots\\
\mathbf{T}^{[k(R+1)]}_R \barxV^{[2]} \end{array} \right] )\subset
\text{span}( \left[ \begin{array}{cccc}\barxV^{[1]}&\mathbf{0}& \cdots &\mathbf{0}\\
\mathbf{0}&\barxV^{[1]}&\cdots&\mathbf{0}\\ \vdots& \vdots & \ddots
&\vdots\\ \mathbf{0}& \mathbf{0} & \cdots & \barxV^{[1]}
\end{array}\right])
\end{equation*}
The above condition can be satisfied if
\begin{equation}
\mathbf{T}^{[k(R+1)]}_i\bar{\mathbf{V}}^{[2]} \prec
\bar{\mathbf{V}}^{[1]} \label{c3}~k=R+2,R+3,\cdots,K~i=1,\cdots,R
\end{equation}

Thus, interference alignment is ensured by choosing
$\bar{\mathbf{V}}^{[1]}$ and $\bar{\mathbf{V}}^{[2]}$ to satisfy
\eqref{c1}, \eqref{c2}, \eqref{c3}. Note that these conditions can
be expressed as
\begin{equation}
\mathbf{T}^{[kj]}_i\bar{\mathbf{V}}^{[2]} \prec
\bar{\mathbf{V}}^{[1]}~~ \forall (k,j)\in A ~i=1,2,\cdots,R
\end{equation}
where $A=\{(k,j):(k,j)\in \{1,2,\cdots,R+1\}\times\{R+2,\cdots,K\}\}
\cup\{(k,j):(k,j)\in\{R+2,\cdots,K\}\times\{R+1,\cdots,K\},~k \neq
j\}$. Therefore, there are $KR(K-R-1)$ such equations. We need to
choose $R(n+1)^{\Gamma}$ column vectors for $\barxV^{[1]}$ and
$Rn^{\Gamma}$ column vectors for $\barxV^{[2]}$. Let $\mathbf{w}$ be
a $\mu_n \times 1$ column vector $\mathbf{w} = [1 \ 1 \ \ldots \
1]^T$.  The sets of column vectors of $\barxV^{[1]}$ and
$\barxV^{[2]}$ are chosen to be equal to the sets $\bar{V}^{[1]}$
and $\bar{V}^{[2]}$ respectively where
\begin{equation}
\bar{V}^{[1]} = \bigcup_{m=0}^{R-1} \big\{ \big(\prod_{i=1,\cdots,R,
(k,j)\in A} (\mathbf{T}^{[kj]}_i)^{\alpha^{[kj]}_i}\big)\mathbf{w}:\
\alpha^{[kj]}_i \in \{mn+m+1, mn+m+2, \ldots, (m+1)n+m+1
\}\big\}\label{v1}
\end{equation}
\begin{equation}
\bar{V}^{[2]} = \bigcup_{m=0}^{R-1} \big\{ \big(\prod_{i=1,\cdots,R,
(k,j)\in A} (\mathbf{T}^{[kj]}_i)^{\alpha^{[kj]}_i}\big)\mathbf{w}:
\alpha_i^{[kj]} \in \{mn+m+1, mn+m+2, \ldots, (m+1)n+m
\}\big\}\label{v2}
\end{equation}
Note that the above construction requires the commutative property
of multiplication of matrices $\mathbf{T}^{[kj]}_i$. Therefore, it
requires $\mathbf{T}^{[kj]}_i$ to be diagonal matrices. Next, we
will show this is true. We illustrate this for the case when
$k=R+2,\cdots,K$ and $j=R+2,\cdots,k-1,k+1,\cdots,K$. Similar
arguments can be applied to other cases. Notice that
$[\bar{\mathbf{H}}^{[k1]}~ \bar{\mathbf{H}}^{[k2]}~\cdots
~\bar{\mathbf{H}}^{[kR]}]$ is a $R\mu_n \times R\mu_n$ square
matrix:
\begin{eqnarray*}
&\left[\begin{array}{cccc}&\bar{\mathbf{H}}^{[k1]}~
\bar{\mathbf{H}}^{[k2]} \cdots
\bar{\mathbf{H}}^{[kR]}\end{array}\right] =\\
&\begin{tiny}\left[ \begin{array}{ccccccccc}  \mathbf{h}^{[k1]}(\mu_n(t-1)+1) & \mathbf{0}_{R \times 1} & \ldots & \mathbf{0}_{R \times 1}&\cdots& \mathbf{h}^{[kR]}(\mu_n(t-1)+1) & \mathbf{0}_{R \times 1} & \ldots & \mathbf{0}_{R \times 1}\\
     \mathbf{0}_{R \times 1} & \mathbf{h}^{[k1]}(\mu_n(t-1)+2) & \ldots & \mathbf{0}_{R \times 1}&\cdots&\mathbf{0}_{R \times 1} & \mathbf{h}^{[kR]}(\mu_n(t-1)+2) & \ldots & \mathbf{0}_{R \times 1}\\
    \vdots & \vdots & \ddots & \vdots &\cdots&\vdots & \vdots & \ddots & \vdots\\
     \mathbf{0}_{R \times 1} & \mathbf{0}_{R \times 1}& \cdots  &
     \mathbf{h}^{[k1]}(\mu_nt)&\cdots& \mathbf{0}_{R \times 1} & \mathbf{0}_{R \times 1}& \cdots  & \mathbf{h}^{[kR]}(\mu_nt)
    \end{array}\right]\end{tiny}
\end{eqnarray*}
Then,
\begin{eqnarray*}
[\bar{\mathbf{H}}^{[k1]}~
\bar{\mathbf{H}}^{[k2]}~\cdots ~\bar{\mathbf{H}}^{[kR]}]^{-1}=\begin{tiny}\left[ \begin{array}{cccc}\mathbf{u}^{[k1]}(\mu_n(t-1)+1)_{1 \times R} & \mathbf{0}_{1 \times R} & \cdots & \mathbf{0}_{1 \times R}\\
      \mathbf{0}_{1 \times R}& \mathbf{u}^{[k1]}(\mu_n(t-1)+2)_{1 \times R} & \cdots & \mathbf{0}_{1 \times R}\\
    \vdots & \vdots & \ddots & \vdots \\
    \mathbf{0}_{1 \times R} & \mathbf{0}_{1 \times R}& \cdots &\mathbf{u}^{[k1]}(\mu_n(t-1)+\mu_n)_{1 \times R}\\ \mathbf{u}^{[k2]}(\mu_n(t-1)+1)_{1 \times R} & \mathbf{0}_{1 \times R} & \cdots & \mathbf{0}_{1 \times R}\\
      \mathbf{0}_{1 \times R}& \mathbf{u}^{[k2]}(\mu_n(t-1)+2)_{1 \times R} & \cdots & \mathbf{0}_{1 \times R}\\
    \vdots & \vdots & \ddots & \vdots \\
    \mathbf{0}_{1 \times R} & \mathbf{0}_{1 \times R}& \cdots &\mathbf{u}^{[k2]}(\mu_n(t-1)+\mu_n)_{1 \times R}\\ \vdots & \vdots & \vdots & \vdots\\
    \mathbf{u}^{[kR]}(\mu_n(t-1)+1)_{1 \times R} & \mathbf{0}_{1 \times R} & \cdots & \mathbf{0}_{1 \times R}\\
      \mathbf{0}_{1 \times R}& \mathbf{u}^{[kR]}(\mu_n(t-1)+2)_{1 \times R} & \cdots & \mathbf{0}_{1 \times R}\\
    \vdots & \vdots & \ddots & \vdots \\
    \mathbf{0}_{1 \times R} & \mathbf{0}_{1 \times R}& \cdots &\mathbf{u}^{[kR]}(\mu_n(t-1)+\mu_n)_{1 \times R}
    \end{array}\right]\end{tiny}
\end{eqnarray*}
where $\mathbf{u}^{[kj]}(\mu_n(t-1)+\kappa), \forall
\kappa=1,2,\ldots,\mu_n$ is a $1 \times R$ row vector and
\begin{eqnarray}
\begin{tiny}\left[ \begin{array}{cccc}\mathbf{h}^{[k1]}(\mu_n(t-1)+\kappa)& \mathbf{h}^{[k2]}(\mu_n(t-1)+\kappa) &
\cdots & \mathbf{h}^{[kR]}(\mu_n(t-1)+\kappa)\end{array} \right]^{-1} = \left[ \begin{array}{c}\mathbf{u}^{[k1]}(\mu_n(t-1)+\kappa)\\\mathbf{u}^{[k2]}(\mu_n(t-1)+\kappa) \notag\\
\vdots\\\mathbf{u}^{[kR]}(\mu_n(t-1)+\kappa)\end{array}\right]~~~~\kappa=1,2,\ldots,\mu_n.\end{tiny}
\end{eqnarray}
Recall
\begin{eqnarray*}
\mathbf{T}^{[kj]}=\left[
\begin{array}{c}\mathbf{T}^{[kj]}_1\\ \mathbf{T}^{[kj]}_2\\\vdots\\
\mathbf{T}^{[kj]}_R\end{array}\right]=[\bar{\mathbf{H}}^{[k1]}~
\bar{\mathbf{H}}^{[k2]}~\cdots
~\bar{\mathbf{H}}^{[kR]}]^{-1}\bar{\mathbf{H}}^{[kj]}~~
  \bar{\mathbf{H}}^{[kj]}(t) =  \begin{tiny}\left[ \begin{array}{cccc}  \mathbf{h}^{[kj]}(\mu_n(t-1)+1) & \mathbf{0} & \ldots & \mathbf{0}\\
     \mathbf{0} & \mathbf{h}^{[kj]}(\mu_n(t-1)+2) & \ldots & \mathbf{0}\\
    \vdots & \vdots & \ddots & \vdots\\
     \mathbf{0} & \mathbf{0}& \cdots  & \mathbf{h}^{[kj]}(\mu_nt)
    \end{array}\right]\end{tiny}
\end{eqnarray*}
Thus, $\forall i=1,2,\cdots,R$
\begin{equation}\label{diagonal}
\mathbf{T}^{[kj]}_i=\begin{tiny}\left[ \begin{array}{cccc}\mathbf{u}^{[ki]}(\mu_n(t-1)+1)\mathbf{h}^{[kj]}(\mu_n(t-1)+1)& 0 & \cdots & 0\\
     0 & \mathbf{u}^{[ki]}(\mu_n(t-1)+2)\mathbf{h}^{[kj]}(\mu_n(t-1)+2) & \cdots & 0\\
    \vdots & \vdots & \ddots & \vdots \\
    0 & 0 & \cdots
    &\mathbf{u}^{[ki]}(\mu_n t)\mathbf{h}^{[kj]}(\mu_n t)\end{array}\right]\end{tiny}
\end{equation}
Hence, $\mathbf{T}^{[kj]}_i$ are diagonal matrices with diagonal
entries
$\mathbf{u}^{[ki]}(\mu_n(t-1)+\kappa)\mathbf{h}^{[kj]}(\mu_n(t-1)+\kappa)$,
$\forall \kappa=1,\ldots,\mu_n $.

Through interference alignment, we ensure that the dimension of the
interference is small enough. Now we need to verify that the desired
signal vectors are linearly independent of the interference vectors
so that each receiver can separate the signal and interference
signals. Consider Receiver 1. Since all interference vectors are
aligned in the signal space spanned by interference from transmitter
$2,3\cdots,R+1$, it suffices to verify that columns of
$\bar{\mathbf{H}}^{[11]}\barxV^{[1]}$ are linearly independent of
columns of  $[\bar{\mathbf{H}}^{[12]}\barxV^{[1]} \cdots
\bar{\mathbf{H}}^{[1(R+1)]}\barxV^{[1]}]$ almost surely. Notice that
the direct channel matrix $\barxH^{[11]}$ does not appear in the
interference alignment equations and $\barxV^{[1]}$ is chosen
independently of $\barxH^{[11]}$. Then, the desired signal
$\barxV^{[1]}$ undergoes an independent linear transformation by
multiplying $\barxH^{[11]}$. Thus, columns of
$\bar{\mathbf{H}}^{[11]}\barxV^{[1]}$ are linearly independent of
columns of $[\bar{\mathbf{H}}^{[12]}\barxV^{[1]} \cdots
\bar{\mathbf{H}}^{[1(R+1)]}\barxV^{[1]}]$ almost surely as long as
all entries of $\barxV^{[1]}$ are not equal to zero with probability
one. If there are some entries of $\barxV^{[1]}$ are equal to zero,
then due to the block diagonal structure of $\barxH^{[11]}$ the
desired signal vectors are linearly dependent of the interference
vectors. For example, consider three $3 \times 3$ diagonal matrix
$\mathbf{H}^{[1]}$, $\mathbf{H}^{[2]}$, $\mathbf{H}^{[3]}$ whose
entries are drawn according to a continuous distribution.
$\mathbf{v}$ is a $3 \times 1$ vector whose entries depend on
entries of $\mathbf{H}^{[2]}$, $\mathbf{H}^{[3]}$ and are non-zero
with probability one. Vectors $\mathbf{H}^{[2]}\mathbf{v}$ and
$\mathbf{H}^{[3]}\mathbf{v}$ span a plane in the three dimensional
space. Now vector $\mathbf{v}$ undergoes a random linear
transformation by multiplying $\mathbf{H}^{[1]}$. The probability
that vector $\mathbf{H}^{[1]}\mathbf{v}$ lies in that plan is zero.
If $\mathbf{v}$ has one zero entry, for example
$\mathbf{v}=[1~1~0]^T$, then $\mathbf{H}^{[1]}\mathbf{v},
\mathbf{H}^{[2]}\mathbf{v}$ and $\mathbf{H}^{[3]}\mathbf{v}$ are two
dimensional vectors in the three dimensional vector space. Hence
they are linearly dependent. Next we will verify all entries of
$\barxV^{[1]}$ and $\barxV^{[2]}$are nonzero with probability one
through their construction from $\eqref{v1}$ and $\eqref{v2}$. From
\eqref{v1}, \eqref{v2} and \eqref{diagonal}, it can be seen that
each entry of $\barxV^{[1]}$ and $\barxV^{[2]}$ is a product of the
power of some
$\mathbf{u}^{[ki]}(\mu_n(t-1)+\kappa)\mathbf{h}^{[kj]}(\mu_n(t-1)+\kappa)$.
To verify each entry of $\barxV^{[1]}$ and $\barxV^{[2]}$ is not
equal to zero with probability one, we only need to verify
$\mathbf{u}^{[ki]}(\mu_n(t-1)+\kappa)\mathbf{h}^{[kj]}(\mu_n(t-1)+\kappa)$
is not equal to zero with probability one. Since each entry of
$\mathbf{h}^{[kj]}(\mu_n(t-1)+\kappa)$ is drawn from a continuous
distribution,
$\mathbf{u}^{[ki]}(\mu_n(t-1)+\kappa)\mathbf{h}^{[kj]}(\mu_n(t-1)+\kappa)=0$
if and only if all entries of $\mathbf{u}^{[ki]}(\mu_n(t-1)+\kappa)$
are equal to zero. However, $\mathbf{u}^{[ki]}(\mu_n(t-1)+\kappa)$
is a row of the inverse of the $R\times R$ square matrix. Thus, not
all entries of $\mathbf{u}^{[ki]}(\mu_n(t-1)+\kappa)$ are equal to
zero with probability one. As a result, all entries of
$\barxV^{[1]}$ and $\barxV^{[2]}$ are not equal to zero with
probability one. To this end, we conclude that at Receiver 1 the
desired signal vectors are linearly independent with the
interference signal vectors.

Similar arguments can be applied at Receiver $2,3,\ldots,K$ to show
that the desired signal vectors are linearly independent of the
interference vectors. Thus, each receiver can decode its desired
streams using zero forcing. As a result, each user can achieve
$\frac{R}{R+1}$ degrees of freedom per channel use for a total of
$\frac{R}{R+1}K$ degrees of freedom with probability one.
\end{proof}

\section{The achievable degrees of freedom of the MIMO Gaussian Interference channel with constant channel
coefficients}\label{apdx:mimo} In this appendix, we consider the
achievable degrees of freedom for some MIMO Gaussian interference
channels with constant channel coefficients. Specifically, we
consider the $R+2$ user MIMO Gaussian interference channel where
each transmitter has $M>1$ antennas and receiver has $RM$,
$R=2,3,\cdots$ antennas. The main results of this section are
presented in the following theorems:
\begin{theorem}\label{theorem:cwot}
For the $R+2$ user MIMO Gaussian interference channel where each
transmitter has $M>1$ antennas and each receiver has $RM$,
$R=2,3,\cdots$, antennas with constant channel coefficients,
$RM+\lfloor\frac{RM}{R^2+2R-1}\rfloor$ degrees of freedom can be
achieved without channel extension.
\end{theorem}
\begin{proof}
The achievable scheme is provided in the following part.
\end{proof}
Theorem \ref{theorem:cwot} is interesting because it shows that when
$\lfloor\frac{RM}{R^2+2R-1}\rfloor>0$ and hence $M>R+2-\frac{1}{R}$,
using interference alignment scheme combined with zero forcing can
achieve more degrees of freedom than merely zero forcing. It also
shows that the $R+2$ user MIMO interference channel with $M$
antennas at each transmitter and $RM$ antennas at each receiver can
achieve more degrees of freedom than $R+1$ user with the same
antenna deployment when $M>R+2-\frac{1}{R}$. For example, if $R=2$,
Theorem \ref{theorem:cwot} shows that for the 4 user interference
channel with $M$ and $2M$ antennas at each transmitter and receiver
respectively, $2M+\lfloor\frac{2M}{7}\rfloor$ degrees of freedom can
be achieved using interference alignment. However, only $2M$ degrees
of freedom can be achieved using zero forcing. Thus, when $M>3$,
using interference alignment combined with zero forcing can achieve
more degrees of freedom than merely zero forcing. Similarly, only
$2M$ degrees of freedom can be achieved on the 3 user interference
channel with the same antenna deployment. Hence, when $M>3$ more
degrees of freedom can be achieved on the 4 user interference
channel. While Theorem~$\ref{theorem:cwot}$ indicates that when
$M<R+2$ using interference alignment combined with zero forcing may
not achieve more degrees of freedom than zero forcing without
channel extension, using interference alignment can achieve more
degrees of freedom if we allow channel extension. We present the
result in the following theorem:
\begin{theorem}\label{theorem:cex}
For the $R+2$ user MIMO interference channel where each transmitter
has $M$ $(1<M<R+2)$ antennas and each receiver has $RM$,
$R=2,3,\cdots$, antennas with constant channel coefficients,
$RM+\frac{1}{\lceil\frac{R+2}{M}\rceil}$ degrees of freedom per
orthogonal dimension can be achieved with
$\lceil\frac{R+2}{M}\rceil$ channel extension.
\end{theorem}
\begin{proof}
The achievable scheme is provided in the following part.
\end{proof}

Theorem $\ref{theorem:cex}$ shows that if we allow channel
extension, $\frac{1}{\lceil\frac{R+2}{M}\rceil}$ more degrees of
freedom can be achieved using interference alignment combined with
zero forcing than merely zero forcing. For example, when $R=2, M=2$,
$\frac{1}{2}$ more degrees of freedom can be achieved using
interference alignment.

\subsection{Proof of Theorem \ref{theorem:cwot}}\label{achievability:t2}
When $\lfloor\frac{RM}{R^2+2R-1}\rfloor<0$ and hence $M
<R+2-\frac{1}{R}$, $RM$ degrees of freedom can be achieved by zero
forcing at each receiver. When $M \geq R+2$, we provide an
achievable scheme based on interference alignment to show that the
$i^{th}$ user can achieve $d_i$ degrees of freedom where
$R\lfloor\frac{R M}{R^2+2R-1}\rfloor \leq d_i \leq M$ and
$d_1+\cdots+d_{R+2}=RM+\lfloor\frac{RM}{R^2+2R-1}\rfloor$.

Transmitter $i$ sends message $W_i$ to Receiver $i$ using $d_i$
independently encoded streams along vectors $\mathbf{v}^{[i]}_m$,
i.e,
\begin{equation*}
\mathbf{X}^{[i]}=\sum_{m=1}^{d_i}x^{i}_m\mathbf{v}_m^{[i]}=\mathbf{V}^{[i]}\mathbf{X}^{i}~i=1,\cdots,R+2
\end{equation*}
Then, the received signal is
\begin{equation*}
\mathbf{Y}^{[j]}=\sum_{i=1}^{R+2}\mathbf{H}^{[ji]}\mathbf{V}^{[i]}\mathbf{X}^i+\mathbf{Z}^{[j]}.
\end{equation*}
In order for each receiver to decode its desired signal streams by
zero forcing the interference, the dimension of the interference has
to be less than or equal to $RM-d_i$. However, there are
$\lfloor\frac{RM}{R^2+2R-1}\rfloor+RM-d_i$ interference vectors at
Receiver $i$. Therefore, we need to align
$\lfloor\frac{RM}{R^2+2R-1}\rfloor$ interference signal vectors at
each receiver. This can be achieved if
$\lfloor\frac{RM}{R^2+2R-1}\rfloor$ interference vectors are aligned
within the space spanned by all other interference vectors. First,
we write $\xV^{[i]}$ in the block matrix form:
\begin{equation*}
\xV^{[i]}=[\xV^{[i]}_1~\xV^{[i]}_2~\cdots~\xV^{[i]}_R~\xV^{[i]}_{R+1}]
\end{equation*}
where $\xV^{[i]}_1,\cdots,\xV^{[i]}_R$ are $M \times
\lfloor\frac{RM}{R^2+2R-1}\rfloor$ dimensional matrices and
$\xV^{[i]}_{R+1}$ is an $M \times (d_i-R
\lfloor\frac{RM}{R^2+2R-1}\rfloor)$ dimensional matrix. At Receiver
1, we align the range of $\mathbf{H}^{[1(R+2)]}\mathbf{V}^{[R+2]}_1$
within the space spanned by other interference vectors:
\begin{eqnarray}
\text{span}(\mathbf{H}^{[1(R+2)]}\mathbf{V}^{[R+2]}_1) \subset
\text{span}(\left[
\mathbf{H}^{[12]}\mathbf{V}^{[2]}~\mathbf{H}^{[13]}\mathbf{V}^{[3]}~\cdots~\mathbf{H}^{[1(R+1)]}\mathbf{V}^{[R+1]}
\right])\notag \\
\Rightarrow
\text{span}(\underbrace{[\mathbf{H}^{[12]}~\mathbf{H}^{[13]}~\cdots~\mathbf{H}^{[1(R+1)]}]^{-1}\mathbf{H}^{[1(R+2)]}}_{\mathbf{T}^{[1]}}\mathbf{V}^{[R+2]}_1)
\subset \text{span}(\left[\begin{array}{cccc}\xV^{[2]}&\mathbf{0}& \cdots &\mathbf{0}\\
\mathbf{0}&\xV^{[3]}&\cdots&\mathbf{0}\\ \vdots& \vdots & \ddots
&\vdots\\ \mathbf{0}& \mathbf{0} & \cdots & \xV^{[R+1]}
\end{array} \right]) \label{eqnrx1}
\end{eqnarray}
Note that $\mathbf{T}^{[1]}$ is a $RM \times M$ matrix and can be
written in a block matrix form:
\begin{equation*}
\mathbf{T}^{[1]}=\left[
\begin{array}{c}\mathbf{T}^{[1]}_1 \\ \mathbf{T}^{[1]}_2 \\\vdots\\
\mathbf{T}^{[1]}_{R} \end{array}\right]
\end{equation*}
Then, condition \eqref{eqnrx1} can be expressed equivalently as
\begin{equation*}
\text{span}(\left[
\begin{array}{c}\mathbf{T}^{[1]}_1 \mathbf{V}^{[R+2]}_1)\\ \mathbf{T}^{[1]}_2 \mathbf{V}^{[R+2]}_1)\\\vdots\\
\mathbf{T}^{[1]}_{R} \mathbf{V}^{[R+2]}_1)\end{array}\right])
\subset \text{span}(\left[\begin{array}{cccc}\xV^{[2]}&\mathbf{0}& \cdots &\mathbf{0}\\
\mathbf{0}&\xV^{[3]}&\cdots&\mathbf{0}\\ \vdots& \vdots & \ddots
&\vdots\\ \mathbf{0}& \mathbf{0} & \cdots & \xV^{[R+1]}
\end{array} \right])
\end{equation*}
This condition can be satisfied if
\begin{eqnarray}
\mathbf{T}^{[1]}_1
\mathbf{V}^{[R+2]}_1&=& \xV^{[2]}_1 \notag\\ \mathbf{T}^{[1]}_2 \mathbf{V}^{[R+2]}_1 &=& \xV^{[3]}_1 \notag\\
& \vdots & \notag \\ \mathbf{T}^{[1]}_{R-1}
\mathbf{V}^{[R+2]}_1 &=& \xV^{[R]}_1 \notag\\
\text{span}(\mathbf{T}^{[1]}_{R}
\mathbf{V}^{[R+2]}_1)&=&\text{span}(\xV^{[R+1]}_1)\label{span1}
\end{eqnarray}
At Receiver 2, we align the range of
$\mathbf{H}^{[2(R+2)]}\mathbf{V}^{[R+2]}_1$ within the space spanned
by other interference vectors:
\begin{equation*}
\text{span}(\mathbf{H}^{[2(R+2)]}\mathbf{V}^{[R+2]}_1) \subset
\text{span}(\left[
\mathbf{H}^{[21]}\mathbf{V}^{[1]}~\mathbf{H}^{[23]}\mathbf{V}^{[3]}~\cdots~\mathbf{H}^{[2(R+1)]}\mathbf{V}^{[R+1]} \right])\\
\end{equation*}
By similar arguments used at Receiver 1, this condition can be
satisfied if
\begin{eqnarray}
\mathbf{T}^{[2]}_1
\mathbf{V}^{[R+2]}_1&=& \xV^{[1]}_1  \notag\\ \mathbf{T}^{[2]}_2 \mathbf{V}^{[R+2]}_1 &=& \xV^{[3]}_2 \notag\\
& \vdots & \notag \\ \mathbf{T}^{[2]}_{R-1}
\mathbf{V}^{[R+2]}_1 &=& \xV^{[R]}_2 \notag\\
\text{span}(\mathbf{T}^{[2]}_{R}
\mathbf{V}^{[R+2]}_1)&=&\text{span}(\xV^{[R+1]}_1)\label{span2}
\end{eqnarray}
where
\begin{equation*}
\mathbf{T}^{[2]}= \left[
\begin{array}{c}\mathbf{T}^{[2]}_1 \\ \mathbf{T}^{[2]}_2 \\ \vdots\\
\mathbf{T}^{[2]}_{R} \end{array}\right]=
[\mathbf{H}^{[21]}~\mathbf{H}^{[23]}~\cdots~\mathbf{H}^{[2(R+1)]}]^{-1}\mathbf{H}^{[2(R+2)]}
\end{equation*}
At Receiver $j$, $\forall j,  2<j\leq R+1$, we align the range of
$\mathbf{H}^{[j(j-1)]}\mathbf{V}^{[j-1]}_1$ within the space spanned
by other interference vectors:
\begin{equation*}
\text{span}(\mathbf{H}^{[j(j-1)]}\mathbf{V}^{[j-1]}_1) \subset
\text{span}(\left[
\barxH^{[j1]}\xV^{[1]}~\cdots~\barxH^{[j(j-2)]}\xV^{[j-2]}~\barxH^{[j(j+1)]}\xV^{[j+1]}~\cdots~\barxH^{[ji]}\xV^{[i]}~\cdots~\barxH^{[j(R+2)]}\xV^{[R+2]}
\right])
\end{equation*}
By similar arguments used at Receiver 1, this condition can be
satisfied if
\begin{eqnarray*}
\mathbf{T}^{[j]}\xV^{[j-1]}_1=\left[\begin{array}{c}\xV^{[1]}_{n(1,j)}\\
\vdots\\ \xV^{[j-2]}_{n(j-2,j)}\\ \xV^{[j+1]}_{n(j+1,j)}\\ \vdots\\
\xV^{[i]}_{n(i,j)}\\ \vdots\\ \xV^{[R+2]}_{n(R+2,j)}
\end{array}\right]
\end{eqnarray*}
where
\begin{equation*}
\mathbf{T}^{[j]}=[\mathbf{H}^{[j1]}~\cdots~\mathbf{H}^{[j(j-2)]}~\mathbf{H}^{[j(j+1)]}~\cdots~\mathbf{H}^{[j(R+2)]}]^{-1}\mathbf{H}^{[j(j-1)]}~~~
n(i,j)=\left\{\begin{array}{ccc}j-1&~&i=1,R+1,R+2,i \neq
j\\j-2&~&1<i<R+1,j>i+1\\j&~&3<i<R+1,j<i\end{array} \right.
\end{equation*}
At Receiver $R+2$, we align the range of
$\mathbf{H}^{[(R+2)1]}\mathbf{V}^{[1]}_1$ within the space spanned
by other interference vectors:
\begin{equation*}
\text{span}(\mathbf{H}^{[(R+2)1]}\mathbf{V}^{[1]}_1) \subset
\text{span}(\left[ \mathbf{H}^{[(R+2)2]}\mathbf{V}^{[2]}~\mathbf{H}^{[(R+2)3]}\mathbf{V}^{[3]}~\cdots~\mathbf{H}^{[(R+2)(R+1)]}\mathbf{V}^{[R+1]}\right])\\
\end{equation*}
This condition can be satisfied if
\begin{eqnarray*}
\mathbf{T}^{[R+2]}\xV^{[1]}_1=\left[\begin{array}{c}\xV^{[2]}_R\\
\xV^{[3]}_R\\ \vdots\\ \xV^{[R+1]}_R \end{array}\right]
\end{eqnarray*}
where
\begin{equation*}
\mathbf{T}^{[R+2]}=[\mathbf{H}^{[(R+2)2]}~\mathbf{H}^{[(R+2)3]}~\cdots~\mathbf{H}^{[(R+2)(R+1)]}]^{-1}\mathbf{H}^{[(R+2)1]}
\end{equation*}
Notice once $\xV^{[R+2]}_1$ is chosen, all other vectors can be
solved from the above equations. To solve $\xV^{[R+2]}_1$, from
\eqref{span1}, \eqref{span2}, we have
\begin{eqnarray*}
\text{span}
(\mathbf{T}^{[1]}_{R}\mathbf{V}^{[R+2]}_1)&=&\text{span}(\mathbf{T}^{[2]}_{R}\mathbf{V}^{[R+2]}_1)\label{eigen}\\
\Rightarrow \text{span}
((\mathbf{T}^{[2]}_{R})^{-1}\mathbf{T}^{[1]}_{R}\mathbf{V}^{[R+2]}_1)&=&\text{span}(\mathbf{V}^{[R+2]}_1)
\end{eqnarray*}
Hence, columns of $\xV^{[R+2]}_1$ can be chosen as
\begin{equation}
\mathbf{V}^{[R+2]}_{1}=[\mathbf{e}_1~\cdots~\mathbf{e}_{\lfloor\frac{RM}{R^2+2R-1}\rfloor}]
\end{equation}
where
$\mathbf{e}_1~\cdots~\mathbf{e}_{\lfloor\frac{RM}{R^2+2R-1}\rfloor}$
are the $\lfloor\frac{RM}{R^2+2R-1}\rfloor$ eigenvectors of
$(\mathbf{T}^{[2]}_{R})^{-1}\mathbf{T}^{[1]}_{R}$. Note that the
above construction only specifies $\mathbf{V}^{[i]}_1,
\mathbf{V}^{[i]}_2,\ldots, \mathbf{V}^{[i]}_R$. The remaining
vectors of $\mathbf{V}^{[i]}_{R+1}$ can be chosen randomly according
to a continuous distribution.

Through interference alignment, we ensure that the interference
vectors span a small enough signal space. We need to verify that the
desired signal vectors, i.e., $\mathbf{H}^{[ii]}\mathbf{V}^{[i]}$
are linearly independent of interference vectors so that each
receiver can decode its message using zero forcing. Notice that the
direct channel matrices $\mathbf{H}^{[ii]}, i=1,\ldots,R+2$ do not
appear in the interference alignment equations, $\mathbf{V}^{[i]}$
undergoes an independent linear transformation by multiplying
$\mathbf{H}^{[ii]}$. Therefore, the desired signal vectors are
linearly independent of the interference signals with probability
one. As a result, user $i$ can achieve $d_i$ degrees of freedom for
a total of $RM+\lfloor\frac{RM}{R^2+2R-1}\rfloor$ degrees of
freedom.

\subsection{Proof of Theorem \ref{theorem:cex}}\label{achievability:t3}
We will provide an achievable scheme based on interference alignment
to show in the $\lceil\frac{R+2}{M}\rceil$ symbol extension channel,
user $i, \forall i=1,3,\ldots,R+2$ can achieve $d_i$ $(R \leq
d_i\leq \lceil\frac{R+2}{M}\rceil M)$ degrees of freedom and user 2
can achieve $d_2$ $(R+1 \leq d_2\leq \lceil\frac{R+2}{M}\rceil M)$
degrees of freedom for a total of $RM\lceil\frac{R+2}{M}\rceil+1$
degrees of freedom. Hence, $RM+\frac{1}{\lceil\frac{R+2}{M}\rceil}$
degrees of freedom can be achieved on the original channel. Over the
extension channel, the channel input-output relationship is
\begin{equation*}
\barxY^{[j]}= \sum_{i=1}^{R+2}\barxH^{[ji]}\barxX^{[i]}+\barxZ^{[j]}
\end{equation*}
where the overbar notation represents the
$\lceil\frac{R+2}{M}\rceil$-symbol extensions so that
\begin{equation*}
\mathbf{\bar{X}}\triangleq\left[\begin{array}{c}\mathbf{X}(\lceil\frac{R+2}{M}\rceil t)\\
\vdots
\\ \mathbf{X}(\lceil\frac{R+2}{M}\rceil
(t+1)-1)\end{array}\right]\quad \mathbf{\bar{Z}}\triangleq\left[\begin{array}{c}\mathbf{Z}(\lceil\frac{R+2}{M}\rceil t)\\
\vdots
\\ \mathbf{Z}(\lceil\frac{R+2}{M}\rceil
(t+1)-1)\end{array}\right]
\end{equation*}
where $\mathbf{X}$ and $\mathbf{Z}$ are $M \times 1$ and $RM \times
1$ vectors respectively, and
\begin{equation*}
\barxH \triangleq \left[\begin{array}{cccc}\xH & \mathbf{0}& \cdots & \mathbf{0}\\
\mathbf{0} & \xH & \cdots& \mathbf{0}\\ \vdots & \vdots & \ddots &
\vdots\\ \mathbf{0} & \mathbf{0} & \cdots & \xH
\end{array}\right].
\end{equation*}
where $\xH$ is the $RM \times M$ channel matrix.

In the extension channel, Transmitter $i$ sends message $W_i$ to
Receiver $i$ using $d_i$ independently encoded streams along vectors
$\barxv^{[i]}_1, \cdots, \barxv^{[i]}_{d_i}$, i.e,
\begin{equation*}
\barxX^{[i]}=\sum_{m=1}^{d_i}\barxv^{[i]}_{m}x^{[i]}_m=\barxV^{[i]}\xX^{[i]}
\end{equation*}
where $\barxV^{[i]}$ and $\xX^{[i]}$ are $M
\lceil\frac{R+2}{M}\rceil \times d_i$ and $d_i \times 1$ matrices
respectively. In order for each receiver to decode its desired
signal streams by zero forcing the interference, the dimension of
the space spanned by the interference vectors has to be less than or
equal to $RM\lceil\frac{R+2}{M}\rceil-d_i$. However, there are
$RM\lceil\frac{R+2}{M}\rceil-d_i+1$ interference vectors at Receiver
$i$. Therefore, we need to align 1 interference signal vector at
each receiver. This can be achieved if one interference vector is
aligned within the space spanned by all other interference vectors.
Mathematically, we choose the
following interference alignment equations:\\
At Receiver 1:
\begin{eqnarray*}
\text{span}(\barxH^{[12]}\barxv^{[2]}_1) \subset
\text{span}(\left[ \barxH^{[13]}\barxV^{[3]}~\barxH^{[14]}\barxV^{[4]}~\cdots~\barxH^{[1(R+1)]}\barxV^{[R+1]}\right])\\
\Rightarrow
\text{span}(\underbrace{[\barxH^{[13]}~\barxH^{[14]}~\cdots~\barxH^{[1(R+1)]}]^{-1}\barxH^{[12]}}_{\mathbf{T^{[1]}}}\barxv^{[2]}_1)
\subset \text{span}(\left[\begin{array}{cccc}\barxV^{[3]}&\mathbf{0}& \cdots &\mathbf{0}\\
\mathbf{0}&\barxV^{[4]}&\cdots&\mathbf{0}\\ \vdots& \vdots & \ddots
&\vdots\\ \mathbf{0}& \mathbf{0} & \cdots & \barxV^{[R+1]}
\end{array} \right])\\
\end{eqnarray*}
This can be achieved if
\begin{eqnarray}\label{s1}
\mathbf{T}^{[1]}\barxv^{[2]}_1=\left[\begin{array}{c}\barxv^{[3]}_1\\
\barxv^{[4]}_1\\ \vdots\\ \barxv^{[R+1]}_1 \end{array}\right]
\end{eqnarray}
At Receiver $j$, $\forall j~2 \leq j\leq R+1$:
\begin{eqnarray*}
&\text{span}(\barxH^{[j(j+1)]}\barxv^{[j+1]}_1) \subset
\text{span}(\left[
\barxH^{[j1]}\barxV^{[1]}~\cdots~\barxH^{[j(j-1)]}\barxV^{[j-1]}~\barxH^{[j(j+2)]}\barxV^{[j+2]}~\cdots~\barxH^{[j(R+2)]}\barxV^{[R+2]}
\right])\\
&\Rightarrow
\text{span}(\underbrace{[\barxH^{[j1]}~\cdots~\barxH^{[j(j-1)]}~\barxH^{[j(j+2)]}~\cdots~\barxH^{[j(R+2)]}]^{-1}\barxH^{[j(j+1)]}}_{\mathbf{T}^{[j]}}\barxv^{[j+1]}_1)
\subset \\ &\text{span}( \begin{tiny} \left[\begin{array}{ccccccc}\barxV^{[1]}&\mathbf{0}& \cdots &\mathbf{0}& \mathbf{0}&\cdots&\mathbf{0}\\ \mathbf{0}&\barxV^{[2]}&\cdots&\mathbf{0}&\mathbf{0}&\cdots&\mathbf{0}\\
\vdots& \vdots & \ddots &\vdots &\vdots &\vdots &\vdots\\
\mathbf{0}&\mathbf{0}& \cdots &\barxV^{[j-1]}& \cdots & \cdots & \mathbf{0}\\
\mathbf{0}&\mathbf{0}& \cdots & \cdots & \barxV^{[j+2]}& \cdots &\mathbf{0}\\
\vdots& \vdots & \vdots &\vdots &\vdots &\ddots &\vdots\\
\mathbf{0}& \mathbf{0} & \cdots &\mathbf{0} &\mathbf{0} &\cdots &
\barxV^{[R+2]}
\\ \end{array} \right]) \end{tiny}
\end{eqnarray*}
This condition can be satisfied if
\begin{eqnarray}\label{s2}
\mathbf{T}^{[j]}\barxv^{[j+1]}_1=\left[\begin{array}{c}\barxv^{[1]}_{n(1,j)}\\
\vdots\\ \barxv^{[j-1]}_{n(j-1,j)}\\ \barxv^{[j+2]}_{n(j+2,j)}\\ \vdots\\
\barxv^{[j]}_{n(i,j)}\\ \vdots\\ \barxv^{[R+2]}_{n(R+2,j)}
\end{array}\right]
\end{eqnarray}
where
\begin{eqnarray*}
n(i,j)=\left\{\begin{array}{ccc}j-1&~&i=1,2,j>i\\j&~&i \geq
3,j<i-1\\j-2&~& i \geq 3, j \geq i+1\end{array} \right.
\end{eqnarray*}
At Receiver $R+2$:
\begin{eqnarray*}
\text{span}(\barxH^{[(R+2)1]}\barxv^{[1]}_1) \subset
\text{span}(\left[ \barxH^{[(R+2)2]}\barxV^{[2]}~\barxH^{[(R+2)3]}\barxV^{[3]}~\cdots~\barxH^{[(R+2)(R+1)]}\barxV^{[R+1]}\right])\\
\Rightarrow
\text{span}(\underbrace{[\barxH^{[(R+2)2]}~\barxH^{[(R+2)3]}~\cdots~\barxH^{[(R+2)(R+1)]}]^{-1}\barxH^{[(R+2)1]}}_{\mathbf{T}^{[R+2]}}\barxv^{[1]}_1)
\subset \text{span}(\left[\begin{array}{cccc}\barxV^{[2]}&\mathbf{0}& \cdots &\mathbf{0}\\
\mathbf{0}&\barxV^{[3]}&\cdots&\mathbf{0}\\ \vdots& \vdots & \ddots
&\vdots\\ \mathbf{0}& \mathbf{0} & \cdots & \barxV^{[R+1]}
\end{array} \right])\\
\end{eqnarray*}
This can be achieved if
\begin{eqnarray}\label{s3}
\mathbf{T}^{[R+2]}\barxv^{[1]}_1=\left[\begin{array}{c}\barxv^{[2]}_{R+1}\\
\barxv^{[3]}_R\\ \vdots\\ \barxv^{[R+1]}_R \end{array}\right]
\end{eqnarray}
Note that once we pick $\barxv^{[2]}_1$, all other vectors can be
solved from \eqref{s1}, \eqref{s2}, \eqref{s3}. $\barxv^{[2]}_1$ can
be chosen randomly according to a continuous distribution as long as
no entry of $\barxv^{[2]}_1$ is equal to zero. Note that the above
construction only specifies
$\barxv^{[i]}_1,\cdots,\barxv^{[i]}_{R}$,$\forall i=1,2,\ldots,R+2$
and $\barxv^{[2]}_{R+1}$. The remaining
$\barxv^{[i]}_{R+1},\cdots,\barxv^{[i]}_{d_i}$,$\forall
i=1,3,\ldots,R+2$ and $\barxv^{[2]}_{R+2},\cdots,\barxv^{[2]}_{d_2}$
can be chosen randomly from a continuous distribution. Since all the
vectors are chosen independently of the direct channel matrices
$\barxH^{[ii]}$ and all entries of $\barxV^{[i]}$ are not equal to
zero almost surely, the desired signal vectors are linearly
independent of the interference vectors at each receiver. As a
result, each receiver can decode its message by zero forcing the
interference to achieve $d_i$ degrees of freedom for a total of
$RM\lceil\frac{R+2}{M}\rceil+1$ degrees of freedom on the
$\lceil\frac{R+2}{M}\rceil$-symbol extension channel. Therefore,
$RM+\frac{1}{\lceil\frac{R+2}{M}\rceil}$ degrees of freedom per
channel use can be achieved on the original channel.

\end{document}